%% file: ms.tex
\documentclass{article}

\usepackage{arxiv}

\usepackage{appendix}
\usepackage{graphicx}
\usepackage{booktabs}
\usepackage{float}
\usepackage{array}
\usepackage{nicematrix}

\usepackage[noend, linesnumbered, ruled]{algorithm2e}
\usepackage{anyfontsize}
\usepackage{t1enc}
\usepackage[upgreek, LGRgreek]{mathastext}
\usepackage[mathscr]{euscript}
\usepackage{doi}

\usepackage{amsthm}

\newtheorem{proposition}{Proposition}
\newtheorem{example}{Example}
\newtheorem{lemma}{Lemma}
\newtheorem{theorem}{Theorem}
\newtheorem{definition}{Definition}
\newtheorem{corollary}{Corollary}
\usepackage{notationbase}

\makeatletter
\newtheorem*{rep@theorem}{\rep@title}
\newcommand{\newreptheorem}[2]{%
\newenvironment{rep#1}[1]{%
 \def\rep@title{#2 \ref{##1}}%
 \begin{rep@theorem}}%
 {\end{rep@theorem}}}
\makeatother

\newreptheorem{theorem}{Theorem}
\newreptheorem{lemma}{Lemma}
\newreptheorem{proposition}{Proposition}

\newcommand{\llIf}[2]{{\let\par\relax\lIf{#1}{#2}}}
\newcommand{\llElse}[1]{{\let\par\relax\lElse{#1}}}

\title{Proof Search on Bilateralist Judgments over Non-deterministic Semantics}

\author{ \href{https://orcid.org/0000-0003-3240-386X}{\includegraphics[scale=0.06]{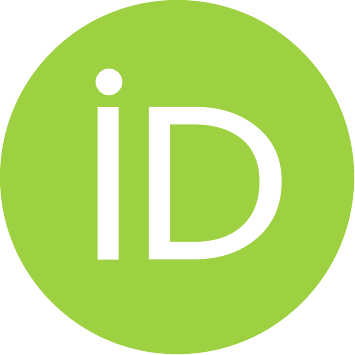}\hspace{1mm}Vitor Greati}\\
Programa de Pós-graduação em Sistemas e Computação (PPgSC)\\
Departamento de Informática e Matemática Aplicada (DIMAp)\\
Universidade Federal do Rio Grande do Norte (UFRN)\\
Natal -- RN, Brazil\\
	\texttt{vitor.greati.017@ufrn.edu.br} \\
	\And
	\href{https://orcid.org/0000-0002-6941-7555}{\includegraphics[scale=0.06]{orcid.pdf}\hspace{1mm}Sérgio Marcelino} \\
SQIG -- Instituto de Telecomunicações\\
Dep. de Matemática -- Instituto Superior Técnico\\
Universidade de Lisboa, Portugal\\
\texttt{smarcel@math.tecnico.ulisboa.pt}\\
    \And
	\href{https://orcid.org/0000-0003-2601-8164}{\includegraphics[scale=0.06]{orcid.pdf}\hspace{1mm}João Marcos} \\
Departamento de Informática e Matemática Aplicada (DIMAp)\\
Universidade Federal do Rio Grande do Norte (UFRN)\\
Natal -- RN, Brazil\\
\texttt{jmarcos@dimap.ufrn.br}
}

\renewcommand{\shorttitle}{Proof Search on Bilateralist Judgments over Non-deterministic Semantics}

\hypersetup{
pdftitle={A template for the arxiv style},
pdfsubject={q-bio.NC, q-bio.QM},
pdfauthor={David S.~Hippocampus, Elias D.~Striatum},
pdfkeywords={First keyword, Second keyword, More},
}

\begin{document}
\maketitle              
\input{tex/abstract}
%
%
%
\input{tex/introduction}
\input{tex/preliminaries}
\input{tex/axiomatizing}

\input{tex/proofsearch}

\input{tex/conclusion}
\paragraph{Acknowledgements} V. Greati and J. Marcos acknowledge support from
CAPES (Brasil) 
--- Finance Code 001 and
CNPq (Brasil), respectively.
S. Marcelino's research was done under the scope of Project UIDB/50008/2020
of Instituto de Telecomunica\c{c}\~oes (IT), financed by the applicable framework
(FCT/MEC through national funds and cofunded by FEDER-PT2020).
%
%
%
\bibliographystyle{plain}
\bibliography{bibliography.bib}
\newpage
\appendix
\input{tex/appendix}

\end{document}

%% file: tex/abstract.tex
\begin{abstract}
\sloppy

The bilateralist approach to logical consequence maintains that judgments of different qualities should be taken into account in determining what-follows-from-what.  We argue that such an approach may be actualized by a two-dimensional notion of entailment induced by semantic structures that also accommodate non-deterministic and partial interpretations, and propose a proof-theoretical apparatus to reason over bilateralist judgments using symmetrical two-dimensional analytical Hilbert-style calculi.  We also provide a proof-search algorithm for finite analytic calculi that runs in at most exponential time, in general, and in polynomial time when only rules having at most one formula in the succedent are present in the concerned calculus.

\keywords{Bilateralism \and Two-dimensional consequence \and Proof Search}
\end{abstract}

%% file: tex/introduction.tex
\section{Introduction}

The conventional approach to bilateralism in logic treats denial as a primitive judgment, on a par with assertion.  One way of allowing these two kinds of judgments to coexist without necessarily allowing them to interfere with one another is by considering a two-dimensional notion of consequence, in which the validity of logical statements obtains in terms of preservation of acceptance along one dimension and of rejection along the other.  From a semantical standpoint, as we will show, this idea may be actualized by the canonical notion of entailment induced by a $\theBPN[\Sigma]$-matrix, a partial non-deterministic logical matrix in which the latter judgments, or cognitive attitudes, are represented by separate collections of truth-values.  This will, in particular, allow for distinct Tarskian (one-dimensional, generalized) consequence relations to coinhabit the same logical structure while keeping their interactions disciplined. 
	
A common practice for incorporating bilateralism into a proof formalism consists in attaching to the underlying formulas a force indicator or signal, say $+$ for assertion and $-$ for denial \cite{rumfitt2000,sergey2019}.  For example, the inference $-(\DefFm\to\DeffFm) \vdash +{\DefFm}$ describes a rule in the bilateral axiomatization of classical logic given in \cite{rumfitt2000}, representing the impossibility of, at once, denying $\DefFm\to\DeffFm$ while failing to assert~$\DefFm$.  In \cite{blasiomarcos2017}, a concurrent approach is offered that consists in working with a two-dimensional notion of consequence, allowing for the cognitive attitudes of acceptance and rejection to act over two separate logical dimensions and taking their interaction into consideration in determining the meaning of logical connectives and of the statements involving them.  The aforementioned inference, for instance, would be expressed by the two-dimensional judgment $\BCon{\varnothing}{\varnothing}{\DefFm\to\DeffFm}{\DefFm}$, which is intended to enforce that an agent is not expected to find reasons for rejecting $\DefFm\to\DeffFm$ while failing to find reasons for accepting $\DefFm$.  From a semantical standpoint, the latter notion of consequence 
may be induced by a two-dimensional logical matrix \cite{blasio20171,blasiomarcos2017}, whose associated two-di\-men\-sion\-al canonical entailment relation very naturally embraces bilateralism and involves two possibly distinct collections of distinguished truth-values: the `designated' values and the `anti-designated' values, respectively equated with acceptance and rejection.

Non-deterministic logical matrices have been extensively investigated in recent years, and proved useful in the construction of effective semantics for many families of logics in a systematic and modular way  \cite{avron2011,MC17,CMocl,MCF20}.  As in \cite{baaz2013}, in the present paper the interpretations of the connectives in a matrix outputs (possibly empty) sets of values, instead of a single value.  In our study, we explore an essential feature of (partial) non-deterministic semantics, namely \emph{effectiveness}, to provide analytic axiomatizations for a very inclusive class of finite \emph{monadic} two-dimensional matrices.  The latter consist in matrices whose underlying linguistic resources are sufficiently expressive so as to uniquely characterize each of the underlying truth-values, in a similar vein as in \cite{marcelinowollic,TCS15-GenComp-CMV}.  In contrast to the multi-dimensional Gentzen-style calculi used in the literature to axiomatize many-valued logics in the context of bilateralism (and multilateralism) \cite{ole2014}, we introduce much simpler two-dimensional symmetrical Hilbert-style calculi to the same effect and show how they give rise to derivations that do not conform to the received view that axiomatic proofs consist simply in `sequences of formulas'.  In our approach, indeed, extending to the bilateralist case the one-dimensional tree-derivation mechanism considered in \cite{shoesmithsmiley1978,MC19synth,marcelinowollic}, the inference rules, instead of manipulating metalinguistic objects, deal only with pairs of accepted / rejected formulas, and derivations are trees whose nodes come labelled with such pairs and result from expansions determined by the rules.  As we will show, the \emph{analyticity} of the axiomatizations that we extract from our two-dimensional (partial) non-deterministic matrices, using symmetrical rules that internalize `case exhaustion', allows for bounded proof search, and the design of a simple recursive decision algorithm that runs in exponential time.  


The paper is organized as follows: 
Section \ref{sec:prelims} introduces the basic concepts and terminology involved in two-dimensional notions of consequence and in symmetrical analytic Hilbert-style calculi.
Section \ref{sec:axiomatization} presents the general axiomatization procedure for finite monadic matrices, illustrating it and highlighting its modularity via the correspondence between refining a matrix and adding rules to a sound symmetrical two-dimensional calculus. 
Then, Section \ref{sec:proofsearch} describes our proposed proof-search algorithm, proves its correctness and investigates its worst-case exponential asymptotic complexity.
In the final remarks, 
we reflect upon the obtained results and indicate some directions for future developments. 
Detailed proofs of the main results may be found at 
\url{https://tinyurl.com/21-GMM-Bilat}.

%% file: tex/preliminaries.tex
\section{Preliminaries}
\label{sec:prelims}

\subsection{Languages}

A \emph{propositional signature}~$\DefSig$ is a family
$\Family{\SigAritySet{\DefSig}{k}}{k \in \NatSet}$, where each $\SigAritySet{\DefSig}{k}$ is a collection of
\emph{$k$-ary connectives}.
Given a denumerable set $\DefProps \SymbDef \Set{\DefProp_i \mid i \in \NatSet}$
of \emph{propositional variables},
the \emph{propositional language over~$\DefSig$ generated by~$\DefProps$},
$\DefLangAlg$, 
is the absolutely free algebra
over~$\DefSig$ freely generated
by~$\DefProps$.
The elements of $\DefLangSet$, the carrier set of the latter algebra, are called
\emph{formulas} and will be indicated
below by capital Roman letters.
As usual, whenever there is no risk of confusion, we will omit braces and unions in collecting sets and formulas, and leave a blank space in place of $\EmptySet$.
For convenience, 
given $\DefFmlaSet \subseteq \DefLangSet$, the set of formulas not in~$\DefFmlaSet$ will be denoted by~$\SetCompl{\DefFmlaSet}$.
On any given language, we may define the functions $\SubfmlasSymb$
and $\PropsSymb$, which output, respectively,
the subformulas and the propositional variables occurring in a given formula, 
and define as well the function $\FmSizeSymb$, such that $\FmSizeSymb(\DefProp) \SymbDef 1$ for each $\DefProp \in \DefProps$, and $\FmSizeSymb(\DefSymbol(\DefFm_1,\ldots,\DefFm_k)) \SymbDef 1 + \Sum_{i=1}^k \FmSizeSymb(\DefFm_i)$, for each $k \in \NatSet$ and $\DefSymbol\in\SigAritySet{\DefSig}{k}$.
Moreover, as usual, endomorphisms on $\DefLangAlg$ are called \emph{substitutions}, 
and, given a formula
$\DeffFm \in \DefLangSet$
with $\Props{\DeffFm} \subseteq \Set{\DefProp_{i_1},\ldots,\DefProp_{i_k}}$, for some $k \in \NatSet$,
we write 
$\DeffFm(\DefFm_1,\ldots, \DefFm_k)$ 
for the image of $\DeffFm$ under a substitution~$\DefSubs$ where
$\DefSubs(\DefProp_{i_j}) = \DefFm_{j}$, for all $1 \leq j \leq k$, and where $\DefSubs(\DefProp) = \DefProp$ otherwise;
for a set~$\DefFmlaSet$ of one-variable formulas, we let $\DefFmlaSet(\DefFm) \SymbDef \Set{\DeffFm(\DefFm) \mid \DeffFm \in \DefFmlaSet}$.

\subsection{Two-dimensional consequence relations}

Hereupon, we shall call \theB\emph{-statement} any 
$2{\times}2$-place tuple 
$\BStat{\DefAccSet}{\DefNRejSet}{\DefRejSet}{\DefNAccSet}$ of sets of formulas in a given language.
By definition, a collection of \theB-statements will be said to constitute a \theB\emph{-con\-se\-quen\-ce relation} $\BConName$
provided that 
any of the following conditions constitutes a sufficient guarantee for the
%
\emph{con\-se\-quen\-ce judgment} $\BCon{\DefAccSet}{\DefNRejSet}{\DefRejSet}{\DefNAccSet}$ to be established:%
\smallskip
\begin{namedproperties}\itemsep1pt
	\item[(O)\label{prop:BConO}] 
			$\DefAccSet \cap \DefNAccSet \neq \EmptySet$
		or 
			$\DefRejSet \cap \DefNRejSet \neq \EmptySet$
	\item[(D)\label{prop:BConD}]
		$\BCon[d]{\DeffAccSet}{\DeffNRejSet}{\DeffRejSet}{\DeffNAccSet}$ and 	$\CogSet{\DeffFmlaSet}{\DefCogVar} \subseteq \CogSet{\DefFmlaSet}{\DefCogVar}$
		for every $\DefCogVar \in \CogsSet$
	\item[(C)\label{prop:BConC}] 
			$\BCon[d]{\IntermAccSet{\CutPropSet}}{\SetCompl{\IntermRejSet{\CutPropSet}}}{\IntermRejSet{\CutPropSet}}{\SetCompl{\IntermAccSet{\CutPropSet}}}$ 
		for all 
			$\DefAccSet \subseteq \IntermAccSet{\CutPropSet} \subseteq \SetCompl{\DefNAccSet}$
		and
			$\DefRejSet \subseteq \IntermRejSet{\CutPropSet} \subseteq \SetCompl{\DefNRejSet}$
	\item[(S)\label{prop:BConS}]
		$\BCon[d]{\DeffAccSet}{\DeffNRejSet}{\DeffRejSet}{\DeffNAccSet}$
		and
		$\CogSet{\DefFmlaSet}{\DefCogVar} = \SubsApply{\DefSubs}{\CogSet{\DeffFmlaSet}{\DefCogVar}}$
		for every $\DefCogVar \in \CogsSet$, for a substitution~$\DefSubs$
		%
\end{namedproperties}
\smallskip
\noindent 
In the above conditions, $\DefAccSet, \DefRejSet, \DefNAccSet, \DefNRejSet$ denote arbitrary sets of formulas, that may intuitively be read as representing, respectively, collections of \emph{accepted}, \emph{rejected}, \emph{non-accepted} and \emph{non-rejected} formulas.
It is not hard to check that such definition, employing the properties of (O)verlap, (D)ilution, (C)ut and (S)ubs\-titution-invariance,
is equivalent to the one found
in~\cite{blasiomarcos2017}, 
and it generalizes the well-known abstract Tarskian one-dimensional account of logical consequence.
In addition, a \theB-consequence relation will be called
\emph{finitary} when a consequence judgment $\BCon{\DefAccSet}{\DefNRejSet}{\DefRejSet}{\DefNAccSet}$ always implies that:
\smallskip
\begin{namedproperties}\itemsep1pt
    \item[(F)\label{prop:BConF}]
    	$\BCon[d]
    	{\FinSubset{\DefAccSet}}
    	{\FinSubset{\DefNRejSet}}
    	{\FinSubset{\DefRejSet}}
    	{\FinSubset{\DefNAccSet}}$, 
	for some finite 
		$\FinSubset{\CogSet{\DefFmlaSet}{\DefCogVar}} \subseteq \CogSet{\DefFmlaSet}{\DefCogVar}$,
	for every $\DefCogVar \in \CogsSet$   
\end{namedproperties}
\smallskip

\noindent We will denote by $\nBConName$
the complement of $\BConName$, sometimes
called the \emph{compatibility relation}
associated to $\BConName$ (cf.~\cite{blasiocaleiromarcos2019}).
Furthermore, we should note that later on we will sometimes write $\InvCog{\Acc}$ for~$\NAcc$, write $\InvCog{\NAcc}$ for~$\Acc$, write $\InvCog{\Rej}$ for~$\NRej$, and write  $\InvCog{\NRej}$ for~$\Rej$.


A \theB-consequence relation $\BConName$ may be said to induce a 2-place relation~$\TConTName$
over $\PowerSet{\DefLangSet}$ by setting $\DefAccSet\SemanticConseq{}{\tAsp}\DefNAccSet$ iff $\BCon{\DefAccSet}{\varnothing}{\varnothing}{\DefNAccSet}$.  This is easily seen to constitute a generalized (one-dimensional) consequence relation.  Another such relation is induced by setting $\DefRejSet\SemanticConseq{}{\fAsp}\DefNRejSet$ iff $\BCon{\varnothing}{\DefNRejSet}{\DefRejSet}{\varnothing}$.  Connected to that, we will say that~$\TConTName$ \emph{inhabits the ${\tAsp}$-aspect of $\BConName$}, and that~$\TConFName$ \emph{inhabits the ${\fAsp}$-aspect of~$\BConName$}.
These are but two of many possible aspects of interest of a given \theB-consequence relation; in principle, very different Tarskian ---and also non-Tarskian!--- logics may coinhabit the same given two-dimensional consequence relation (see \cite{blasiomarcos2017}).

Finally, a \theB-consequence $\BConName$ is said to be \emph{decidable} when there is some \emph{decision procedure} that takes a \theB-statement 
$\BStat{\DefAccSet}{\DefNRejSet}{\DefRejSet}{\DefNAccSet}$
with finite component sets as input, outputs \texttt{true} when $\BCon{\DefAccSet}{\DefNRejSet}{\DefRejSet}{\DefNAccSet}$ is the case, and outputs \texttt{false}
when $\nBCon{\DefAccSet}{\DefNRejSet}{\DefRejSet}{\DefNAccSet}$.

\subsection{Two-dimensional non-deterministic matrices}

A \emph{partial non-deterministic \theB--matrix~$\DefBMatrix$ over a signature} $\DefSig$, or simply $\theBPN[\DefSig]$\emph{--matrix}, is a structure
$\DefBMatrixStruct{\DefBMatrix}$
where the set $\DefVSet$ is said to contain \emph{truth-values}, the sets $\DefDesSet, \DefAntiDesSet \subseteq \DefVSet$ are said to contain, respectively, the \emph{designated} and
the \emph{anti-designated} truth-values, 
and, for each $k \in \NatSet$ and $\DefSymbol \in \SigAritySet{\DefSig}{k}$, the mapping
$\BMatrixInterp{\DefSymbol}{\DefBMatrix} : ({\DefVSet})^k \to \PowerSet{\DefVSet}$
is the \emph{interpretation} of
$\DefSymbol$ in $\DefBMatrix$.
For convenience, we define
$\NDesSet{\DefBMatrix} \SymbDef \SetDiff{\DefVSet}{\DesSet{\DefBMatrix}}$
and
$\NAntiDesSet{\DefBMatrix} \SymbDef \SetDiff{\DefVSet}{\AntiDesSet{\DefBMatrix}}$. A $\theBPN[\DefSig]$--matrix is said to be
\emph{total} when $\EmptySet$ is not in the range of the interpretation of any connective of~$\DefSig$,
\emph{deterministic} when the range of
any interpretation contains only singletons, also called \emph{deterministic images},
%
and \emph{fully indeterministic} if 
it allows for the maximum degree
of non-determinism, that is, if
$\BMatrixInterp{\DefSymbol}{\DefBMatrix}(\DefValueVar_1,\ldots,\DefValueVar_k) = \VSet{\DefBMatrix}$ for each $k \in \NatSet$ and $\DefSymbol \in \SigAritySet{\DefSig}{k}$, and all
$\DefValueVar_1,\ldots,\DefValueVar_k \in \VSet{\DefBMatrix}$.

In the following definitions,~$\DefBMatrix$ will represent an arbitrary $\theBPN[\DefSig]$\emph{--matrix}.

Given a set of truth-values $\DefSetValueVar \subseteq \DefVSet$, the \emph{sub}--$\theBPN[\DefSig]$--\emph{matrix $\SubBMatrix{\DefBMatrix}{\DefSetValueVar}$ induced by $\DefSetValueVar$}
is the $\theBPN[\DefSig]$--matrix $\BMatrixStruct{\SubBMatrix{\DefBMatrix}{\DefSetValueVar}}{\DefSetValueVar}{\DefDesSet \cap \DefSetValueVar}{\DefAntiDesSet \cap \DefSetValueVar}$
such that 
$\BMatrixInterp{\DefSymbol}{\SubBMatrix{\DefBMatrix}{\DefSetValueVar}}(\DefValueVar_1,\ldots,\DefValueVar_k) \SymbDef \BMatrixInterp{\DefSymbol}{\DefBMatrix}(\DefValueVar_1,\ldots,\DefValueVar_k) \cap \DefSetValueVar$, for all $\DefValueVar_1,\ldots,\DefValueVar_k \in \DefSetValueVar$, $k \in \NatSet$ and $\DefSymbol \in \SigAritySet{\DefSig}{k}$.
The set of all subsets of the values of each non-empty total sub--$\theBPN[\DefSig]$--matrix of $\DefBMatrix$ will be denoted by $\TotalSubMVSet{\DefBMatrix}$, that is,
\[
\TotalSubMVSet{\DefBMatrix} \SymbDef \bigcup\limits_{\substack{\EmptySet \neq \DefSetValueVar \subseteq  \DefVSet \\ \SubBMatrix{\DefBMatrix}{ \DefSetValueVar} \text{ total}}} \PowerSet{\DefSetValueVar}.
\]
Check Example \ref{ex:kleenemat} for an illustration
of the latter.

\sloppy We shall call \emph{$\DefBMatrix$-valuation}
any mapping $\DefVal : \DefLangSet \to \DefVSet$ such that
$\DefVal(\DefSymbol(\DefFm_1, \ldots, \DefFm_k)) \in \BMatrixInterp{\DefSymbol}{\DefBMatrix}(
\DefVal(\DefFm_1), \ldots, \DefVal(\DefFm_k))$
 for all $k \in \NatSet$, $\DefSymbol \in \SigAritySet{\DefSig}{k}$
and
$\DefFm_1, \ldots, \DefFm_k \in \DefLangSet$.
As proved in 
\cite{baaz2013}, 
given a set $\DefFmlaSet \subseteq \DefLangSet$ closed under subformulas, 
any mapping $f : \DefFmlaSet \to \DefVSet$
extends to an $\DefBMatrix$-valuation
provided that
$f(\DefSymbol(\DefFm_1, \ldots, \DefFm_k)) \in \BMatrixInterp{\DefSymbol}{\DefBMatrix}(
f(\DefFm_1), \ldots, f(\DefFm_k))$,
for every $\DefSymbol(\DefFm_1, \ldots, \DefFm_k) \in \DefFmlaSet$, 
and $f(\DefFmlaSet) \in \TotalSubMVSet{\DefBMatrix}$.
Notice that if we disregard the latter condition we obtain
the property
of \emph{effectiveness} for total
non-deterministic matrices (\cite{avron2009}); 
as this very condition holds for all such matrices, by making it explicit in the previous definition we obtain a generalization of effectiveness that also applies to partial non-deterministic matrices.
Any formula $\DefFm \in \DefLangSet$
with $\Props{\DefFm} = \Set{\DefProp_{i_1}, \ldots, \DefProp_{i_k}}$
may be interpreted on~$\DefBMatrix$  
as a $k$-ary mapping $\BMatrixInterpFm{\DefFm}{\DefBMatrix}$ such that
$\BMatrixInterpFm{\DefFm}{\DefBMatrix}(\DefValueVar_1,\ldots,\DefValueVar_k) \SymbDef \{ \DefVal(\DefFm) \mid \DefVal \text{ is an $\DefBMatrix$-valuation and }
\DefVal(\DefProp_{i_1}) = \DefValueVar_{1}, \ldots, \DefVal(\DefProp_{i_k}) = \DefValueVar_k\}$.


The \theB-\emph{entailment relation induced by
$\DefBMatrix$}
is a $2{\times}2$-place relation 
$\BEnt{\PlaceholderArg}{\PlaceholderArg}{\PlaceholderArg}{\PlaceholderArg}{\DefBMatrix}$
over $\DefLangSet$
such that:
\begin{itemize}[leftmargin=15mm]    
\item[(\theB-ent)\label{prop:BentDef}]
    \begin{tabular}{lcl}
    $\BEnt[d]{\DefAccSet}{\DefNRejSet}{\DefRejSet}
    {\DefNAccSet}{\DefBMatrix}$
        & \ \ iff\ \ &
        \begin{tabular}{l}
    there is no $\DefBMatrix$-valuation
    $\DefVal$ such that
\\
        $\DefVal(\CogSet{\DefFmlaSet}{\DefCogVar})
        \subseteq \BMatrixDistSet{\DefCogVar}{\DefBMatrix}$
        for every $\DefCogVar \in \CogsSet$
        \end{tabular}
    \end{tabular}
\end{itemize}
\noindent 
for every $\DefAccSet, \DefRejSet, \DefNAccSet, \DefNRejSet \subseteq \DefLangSet$.
Whenever $\BEnt[]{\DefAccSet}{\DefNRejSet}{\DefRejSet}{\DefNAccSet}{\DefBMatrix}$, we say that
the \theB-statement $\BStat[]{\DefAccSet}{\DefNRejSet}{\DefRejSet}{\DefNAccSet}$
\emph{holds} in $\DefBMatrix$.
It is straightforward to check that (see \cite{blasio20171}):%

\begin{proposition}
\label{propo:bentailmentbconseq}
The {{\upshape}\theB}-entailment relation induced by a \upshape{${\theBPN[\Sigma]}$}--matrix is a {\upshape\theB}-con\-se\-quen\-ce relation. 
\end{proposition}

\begin{example}
\label{ex:infomatrix}
Let $\FourValuesSetName \SymbDef \Set{\fVal, \BotVal, \TopVal, \tVal}$,
$\DesSetFour \SymbDef \Set{\TopVal, \tVal}$, $\RejSetFour \SymbDef \Set{\TopVal, \fVal}$,
and consider
a signature $\SigFDE$ containing
but two binary connectives,
$\land$ and $\lor$, and 
one unary connective, $\neg$.
Next, define the $\theBPN[\SigFDE]$--matrix
$\InfoFourMatrix \SymbDef \BMatrixStruct{\InfoFourMatrix}{\FourValuesSetName}{\DesSetFour}{\RejSetFour}$ that interprets the latter connectives according to the  following (non-deterministic) truth-tables (here and below, braces will be omitted from the images of the interpretations):%
\begin{table}[H]
    \centering
    \begin{tabular}{>{\centering\arraybackslash}m{.7cm}!{\vrule width .6pt}>{\centering\arraybackslash}m{.7cm}>{\centering\arraybackslash}m{.7cm}>{\centering\arraybackslash}m{.7cm}>{\centering\arraybackslash}m{.7cm}}
         $\BMatrixInterp{\land}{\InfoFourMatrix}$ & $\fVal$ & $\BotVal$ & $\TopVal$ & $\tVal$\\
         \midrule[.6pt]
         $\fVal$ &  \fVal & \fVal &\fVal&\fVal\\
         $\BotVal$ &\fVal& \fVal, \BotVal & \fVal & \fVal, \BotVal\\
         $\TopVal$ &\fVal& \fVal & \TopVal& \TopVal\\
         $\tVal$ &  \fVal& \fVal, \BotVal & \TopVal & \tVal, \TopVal
    \end{tabular} \quad
    \begin{tabular}{>{\centering\arraybackslash}m{.7cm}!{\vrule width .6pt}>{\centering\arraybackslash}m{.7cm}>{\centering\arraybackslash}m{.7cm}>{\centering\arraybackslash}m{.7cm}>{\centering\arraybackslash}m{.7cm}}
         $\BMatrixInterp{\lor}{\InfoFourMatrix}$ & $\fVal$ & $\BotVal$ & $\TopVal$ & $\tVal$\\
         \midrule[.6pt]
         $\fVal$ &  \fVal, \TopVal & \tVal, \BotVal &\TopVal&\tVal\\
         $\BotVal$ &\tVal, \BotVal& \tVal, \BotVal & \tVal & \tVal\\
         $\TopVal$ &\TopVal& \tVal & \TopVal& \tVal\\
         $\tVal$ &  \tVal& \tVal & \tVal & \tVal
    \end{tabular}\quad
    \begin{tabular}{>{\centering\arraybackslash}m{.7cm}|>{\centering\arraybackslash}m{.7cm}}
          & $\BMatrixInterp{\neg}{\InfoFourMatrix}$ \\
         \midrule[.6pt]
         $\fVal$ & \tVal\\
         $\BotVal$ & \BotVal\\
         $\TopVal$ & \TopVal\\
         $\tVal$ & \fVal
    \end{tabular}
\end{table}%
\noindent 
The $\tAsp$-aspect of $\BEnt{\PlaceholderArg}{\PlaceholderArg}{\PlaceholderArg}{\PlaceholderArg}{\InfoFourMatrix}$ is inhabited by the 
logic introduced in~\cite{avron2007}, which incorporates some principles on how a processor would be expected to deal with information about an arbitrary set of formulas.
\end{example}
Given two $\theBPN[\Sigma]$--matrices
$\DefBMatrix{}_1$ and $\DefBMatrix{}_2$,
we say that $\DefBMatrix{}_2$
is a \emph{refinement} of
$\DefBMatrix{}_1$ when
$\VSet{\DefBMatrix{}_2} \subseteq \VSet{\DefBMatrix{}_1}$ and
$\BMatrixInterp{\DefSymbol}{\DefBMatrix{}_2}(\DefValueVar_1,\ldots,\DefValueVar_k) \subseteq \BMatrixInterp{\DefSymbol}{\DefBMatrix{}_1}(\DefValueVar_1,\ldots,\DefValueVar_k)$ for each $k \in \NatSet$ and $\DefSymbol \in \SigAritySet{\DefSig}{k}$, and for every
$\DefValueVar_1,\ldots,\DefValueVar_k \in \VSet{\DefBMatrix{}_2}$. 
Also, we say that
$\BMatrixInterpPlaceholder{{\DefBMatrix}_2}$ \emph{agrees with}
$\BMatrixInterpPlaceholder{{\DefBMatrix}_1}$ 
when both provide
the same interpretations for
the connectives of~$\DefSig$.
Evidently, every $\theBPN[\Sigma]$--matrix
is a refinement of the corresponding fully indeterministic $\theBPN[\Sigma]$--matrix. In the examples that follow, we
illustrate a couple of refinements of the $\theBPN[\Sigma]$--matrix $\InfoFourMatrix$
presented in Example \ref{ex:infomatrix}, giving rise to
(two-dimensional versions of) other well-known logics.%

\begin{example}
\label{ex:fde}
Let
$\FDEMatrix \SymbDef \BMatrixStruct{\FDEMatrix}{\FourValuesSetName}{\DesSetFour}{\RejSetFour}$ be the $\theBPN[\SigFDE]$--matrix
consisting of a refinement
of $\InfoFourMatrix$ with interpretations given by the following tables:
\begin{table}[H]
    \centering
    \begin{tabular}{>{\centering\arraybackslash}m{.7cm}!{\vrule width .6pt}>{\centering\arraybackslash}m{.7cm}>{\centering\arraybackslash}m{.7cm}>{\centering\arraybackslash}m{.7cm}>{\centering\arraybackslash}m{.7cm}}
         $\BMatrixInterp{\land}{\FDEMatrix}$ & $\fVal$ & $\BotVal$ & $\TopVal$ & $\tVal$\\
         \midrule[.6pt]
         $\fVal$ &  \fVal & \fVal &\fVal&\fVal\\
         $\BotVal$ &\fVal& \BotVal & $\fVal$ & \BotVal\\
         $\TopVal$ &\fVal& $\fVal$ & \TopVal& \TopVal\\
         $\tVal$ &  \fVal& \BotVal & \TopVal & \tVal\\
    \end{tabular}\quad
    \begin{tabular}{>{\centering\arraybackslash}m{.7cm}!{\vrule width .6pt}>{\centering\arraybackslash}m{.7cm}>{\centering\arraybackslash}m{.7cm}>{\centering\arraybackslash}m{.7cm}>{\centering\arraybackslash}m{.7cm}}
         $\BMatrixInterp{\lor}{\FDEMatrix}$ & $\fVal$ & $\BotVal$ & $\TopVal$ & $\tVal$\\
         \midrule[.6pt]
         $\fVal$ &  \fVal & \BotVal &\TopVal&\tVal\\
         $\BotVal$ &\BotVal& \BotVal & $\tVal$ & \tVal\\
         $\TopVal$ &\TopVal& $\tVal$ & \TopVal& \tVal\\
         $\tVal$ &  \tVal& \tVal & \tVal & \tVal\\
    \end{tabular}\quad
    \begin{tabular}{>{\centering\arraybackslash}m{.7cm}|>{\centering\arraybackslash}m{.7cm}}
          & $\BMatrixInterp{\neg}{\FDEMatrix}$ \\
         \midrule[.6pt]
         $\fVal$ & \tVal\\
         $\BotVal$ & \BotVal\\
         $\TopVal$ & \TopVal\\
         $\tVal$ & \fVal\\
    \end{tabular}
\end{table}
\noindent One may readily see that these 
interpretations correspond to the ones of
First Degree Entailment and that
this $\theBPN[\SigFDE]$--matrix corresponds to the logic~$\mathbf{E}^B$ 
presented in \cite{blasio20171}.
\end{example}

\begin{example}
\label{ex:kleenemat}
\sloppy
We may still refine
$\FDEMatrix$ (and thus $\InfoFourMatrix$) a little more.
Let 
$\ImpFreeKleeneMatrix \SymbDef \BMatrixStruct{\ImpFreeKleeneMatrix}{\FourValuesSetName}{\DesSetFour}{\RejSetFour}$ be the $\theBPN[\SigFDE]$--matrix
such that $\BMatrixInterpPlaceholder{\ImpFreeKleeneMatrix}$ agrees with $\BMatrixInterpPlaceholder{\FDEMatrix}$
except that
$\BMatrixInterp{\land}{\ImpFreeKleeneMatrix}(\TopVal, \BotVal)  =\BMatrixInterp{\lor}{\ImpFreeKleeneMatrix}(\TopVal, \BotVal) 
= \BMatrixInterp{\land}{\ImpFreeKleeneMatrix}(\TopVal, \BotVal)
= \BMatrixInterp{\land}{\ImpFreeKleeneMatrix}(\BotVal, \TopVal) = \EmptySet$,
as the following tables show:
\begin{table}[H]
    \centering
    \begin{tabular}{>{\centering\arraybackslash}m{.7cm}!{\vrule width .6pt}>{\centering\arraybackslash}m{.7cm}>{\centering\arraybackslash}m{.7cm}>{\centering\arraybackslash}m{.7cm}>{\centering\arraybackslash}m{.7cm}}
         $\BMatrixInterp{\land}{\ImpFreeKleeneMatrix}$ & $\fVal$ & $\BotVal$ & $\TopVal$ & $\tVal$\\
         \midrule[.6pt]
         $\fVal$ &  \fVal & \fVal &\fVal&\fVal\\
         $\BotVal$ &\fVal& \BotVal & $\EmptySet$ & \BotVal\\
         $\TopVal$ &\fVal& $\EmptySet$ & \TopVal& \TopVal\\
         $\tVal$ &  \fVal& \BotVal & \TopVal & \tVal\\
    \end{tabular}\quad
    \begin{tabular}{>{\centering\arraybackslash}m{.7cm}!{\vrule width .6pt}>{\centering\arraybackslash}m{.7cm}>{\centering\arraybackslash}m{.7cm}>{\centering\arraybackslash}m{.7cm}>{\centering\arraybackslash}m{.7cm}}
         $\BMatrixInterp{\lor}{\ImpFreeKleeneMatrix}$ & $\fVal$ & $\BotVal$ & $\TopVal$ & $\tVal$\\
         \midrule[.6pt]
         $\fVal$ &  \fVal & \BotVal &\TopVal&\tVal\\
         $\BotVal$ &\BotVal& \BotVal & $\EmptySet$ & \tVal\\
         $\TopVal$ &\TopVal& $\EmptySet$ & \TopVal& \tVal\\
         $\tVal$ &  \tVal& \tVal & \tVal & \tVal\\
    \end{tabular}\quad
    \begin{tabular}{>{\centering\arraybackslash}m{.7cm}|>{\centering\arraybackslash}m{.7cm}}
          & $\BMatrixInterp{\neg}{\ImpFreeKleeneMatrix}$ \\
         \midrule[.6pt]
         $\fVal$ & \tVal\\
         $\BotVal$ & \BotVal\\
         $\TopVal$ & \TopVal\\
         $\tVal$ & \fVal\\
    \end{tabular}
\end{table}
\noindent 
Note that $\TotalSubMVSet{\ImpFreeKleeneMatrix} = \Set{\DefSetValueVar \subseteq \FourValuesSetName \mid \Set{\TopVal,\BotVal} \not\subseteq \DefSetValueVar}$. 
As shown in \cite{marcelinowollic}, 
Kleene's strong three-valued logic 
inhabits the 
$\tAsp$-aspect 
of $\BEnt{\PlaceholderArg}{\PlaceholderArg}{\PlaceholderArg}{\PlaceholderArg}{\ImpFreeKleeneMatrix}$.
\end{example}

\begin{example}
\label{ex:mci}
Let $\FiveValuesSetName \SymbDef \Set{\mcif,\mciF, \mciI, \mciT, \mcit}$,
$\DesSetFive \SymbDef \Set{\mciT, \mciI, \mcit}$, $\RejSetFive \SymbDef \Set{\mciT, \mciI, \mcif}$,
and consider
a signature $\SigMCI$ containing
but three binary connectives,
$\land$, $\lor$ and $\mciImp$, and 
two unary connectives, $\neg$ and $\circ$.
Inspired by the 5-valued
non-deterministic logical
matrix presented in \cite{avron2008}
for the logic of formal inconsistency called \textbf{mCi} \cite{marcos-PTS4WCBPL}, we define
the $\theBPN[\SigMCI]$--matrix
$\MCIFiveVMatrix \SymbDef \BMatrixStruct{\MCIFiveVMatrix}{\FiveValuesSetName}{\DesSetFive}{\RejSetFive}$
with the following interpretations:
\begin{gather*}
    \BMatrixInterp{\land}{\MCIFiveVMatrix}(\DefValueVar_1,\DefValueVar_2) \SymbDef 
    \begin{cases}
    \Set{\mcif} & \text{ if either $\DefValueVar_1 \not\in \DesSetFive$ or $\DefValueVar_2 \not\in \DesSetFive$}\\
    \Set{\mcit,\mciI} & \text{ otherwise}
    \end{cases}\\
    \BMatrixInterp{\lor}{\MCIFiveVMatrix}(\DefValueVar_1,\DefValueVar_2) \SymbDef 
    \begin{cases}
    \Set{\mcit,\mciI} & \text{ if either $\DefValueVar_1 \in \DesSetFive$ or $\DefValueVar_2 \in \DesSetFive$}\\
    \Set{\mcif} & \text{if $\DefValueVar_1$,$\DefValueVar_2 \not\in \DesSetFive$}
    \end{cases}\\
    \BMatrixInterp{\mciImp}{\MCIFiveVMatrix}(\DefValueVar_1,\DefValueVar_2) \SymbDef 
    \begin{cases}
    \Set{\mcit,\mciI} & \text{ if either $\DefValueVar_1 \not\in \DesSetFive$ or $\DefValueVar_2 \in \DesSetFive$}\\
    \Set{\mcif} & \text{if $\DefValueVar_1 \in \DesSetFive$  and $\DefValueVar_2 \not\in \DesSetFive$}
    \end{cases}\\
    \begin{tabular}{>{\centering\arraybackslash}m{.7cm}|>{\centering\arraybackslash}m{.7cm}|>{\centering\arraybackslash}m{.7cm}|>{\centering\arraybackslash}m{.7cm}|>{\centering\arraybackslash}m{.7cm}|>{\centering\arraybackslash}m{.7cm}}
          & \mcif&\mciF&\mciI&\mciT&\mcit \\
         \midrule
         $\BMatrixInterp{\neg}{\MCIFiveVMatrix}$&\mcit,\mciI&\mciT&\mcit,\mciI&\mciF&\mcif\\
    \end{tabular}\hspace{.5cm}    
    \begin{tabular}{>{\centering\arraybackslash}m{.7cm}|>{\centering\arraybackslash}m{.7cm}|>{\centering\arraybackslash}m{.7cm}|>{\centering\arraybackslash}m{.7cm}|>{\centering\arraybackslash}m{.7cm}|>{\centering\arraybackslash}m{.7cm}}
          & \mcif&\mciF&\mciI&\mciT&\mcit \\
         \midrule
         $\BMatrixInterp{\circ}{\MCIFiveVMatrix}$&\mciT&\mciT&\mciF&\mciT&\mciT\\
    \end{tabular}
\end{gather*}
\noindent We note that the logic \textbf{mCi} inhabits the $\tAsp$-aspect of $\BEnt{\PlaceholderArg}{\PlaceholderArg}{\PlaceholderArg}{\PlaceholderArg}{\MCIFiveVMatrix}$. 
It is worth pointing out that, up to now, 
no \emph{finite} Hilbert-style calculus was known
to axiomatize this logic;
however, a finite two-dimensional symmetrical
Hilbert-style calculus for \textbf{mCi} results
smoothly from the procedure described
in the next section.
\end{example}

Given 
$\DefSetValueVar, \DeffSetValueVar \subseteq \DefVSet$ and $\DefCogVar \in \Set{\Acc,\Rej}$,
we say that
$\DefSetValueVar$
and $\DeffSetValueVar$
are \emph{$\DefCogVar$-separated},
denoted by
$\SeparatedValues{\DefSetValueVar}{\DeffSetValueVar}{\DefCogVar}$,
if 
$\DefSetValueVar \subseteq \BMatrixDistSet{\DefCogVar}{\DefBMatrix}$
and 
$\DeffSetValueVar \subseteq \SetDiff{\DefVSet}{\BMatrixDistSet{\DefCogVar}{\DefBMatrix}}$,
or vice-versa.
Given two truth-values $\DefValueVar, \DeffValueVar \in \DefVSet$,
a single-variable formula $\DefSeparator$
is a \emph{monadic separator for~$\DefValueVar$ and~$\DeffValueVar$}
whenever
$\SeparatedValues{\BMatrixInterpFm{\DefSeparator}{\DefBMatrix}(\DefValueVar)}{\BMatrixInterpFm{\DefSeparator}{\DefBMatrix}(\DeffValueVar)}{\DefCogVar}$, for some $\DefCogVar \in \{\Acc,\Rej\}$.
The $\theBPN[\Sigma]$--matrix $\DefBMatrix$
is said to be
\emph{monadic} when
for
each pair 
of distinct truth-values of
$\DefBMatrix$
there is a monadic separator
for these values.\footnote{Whether monadicity of a $\theBPN[\Sigma]$--matrix is decidable is still an open problem.}
We say that a set of single-variable formulas $\DiscriminatorVal{}{}{\DefValueVar}{\DefBMatrix}$ \emph{isolates} $\DefValueVar$
whenever, for every $\DeffValueVar\neq\DefValueVar$,
there exists a monadic separator $\DefSeparator \in \DiscriminatorVal{}{}{\DefValueVar}{\DefBMatrix}$ for~$\DefValueVar$ and~$\DeffValueVar$.
A \emph{discriminator for $\DefBMatrix$}, then,
is a family 
$\Discriminator{\DefBMatrix} \SymbDef
\Family{\Tuple{
\DiscriminatorVal{\Acc}{\Acc}{\DefValueVar}{\DefBMatrix},
\DiscriminatorVal{\Acc}{\NAcc}{\DefValueVar}{\DefBMatrix},
\DiscriminatorVal{\Rej}{\Rej}{\DefValueVar}{\DefBMatrix},
\DiscriminatorVal{\Rej}{\NRej}{\DefValueVar}{\DefBMatrix}
}}{\DefValueVar \in \DefVSet}$
such that $\DiscriminatorVal{}{}{\DefValueVar}{\DefBMatrix} \SymbDef \bigcup_{\DefCogVar} \DiscriminatorVal{\DefCogVar}{\DefCogVar}{\DefValueVar}{\DefBMatrix}$
isolates $\DefValueVar$
and $\BMatrixInterpFm{\DefSeparator}{\DefBMatrix}(\DefValueVar) \subseteq \BMatrixDistSet{\DefCogVar}{\DefBMatrix}$ whenever $\DefSeparator \in \DiscriminatorVal{\DefCogVar}{\DefCogVar}{\DefValueVar}{\DefBMatrix}$.
We denote the set $\bigcup_{\DefValueVar \in \DefVSet} \DiscriminatorVal{}{}{\DefValueVar}{\DefBMatrix}$ by $\SeparatorsSet{\DefBMatrix}$
and say that $\Discriminator{\DefBMatrix}$
\emph{is based on}~$\SeparatorsSet{\DefBMatrix}$.

\begin{example}
\label{ex:disc}
The tables below describe, respectively, a discriminator
based on $\Set{\DefProp}$
for any $\theBPN[\Sigma]$--matrix of the form 
$\BMatrixStruct{}{\FourValuesSetName}{\DesSetFour}{\RejSetFour}$ (see Examples \ref{ex:infomatrix}, \ref{ex:fde} and \ref{ex:kleenemat}) and a discriminator for $\MCIFiveVMatrix$ based on $\Set{\DefProp,\neg\DefProp}$ (of Example \ref{ex:mci}):
\begin{table}[H]
    \centering
    \begin{tabular}{>{\centering\arraybackslash}m{.7cm}!{\vrule width .6pt}>{\centering\arraybackslash}m{.7cm}>{\centering\arraybackslash}m{.7cm}>{\centering\arraybackslash}m{.7cm}>{\centering\arraybackslash}m{.7cm}}
        \toprule
         $\DefValueVar$ &$\DiscriminatorVal{\Acc}{\Acc}{\DefValueVar}{\DefBMatrix}$& $\DiscriminatorVal{\Acc}{\NAcc}{\DefValueVar}{\DefBMatrix}$&
         $\DiscriminatorVal{\Rej}{\Rej}{\DefValueVar}{\DefBMatrix}$&
         $\DiscriminatorVal{\Rej}{\NRej}{\DefValueVar}{\DefBMatrix}$\\
         \midrule
         \fVal& $\EmptySet$& $\DefProp$ &$\DefProp$&$\EmptySet$\\
         \BotVal& $\EmptySet$ &$\DefProp$ &$\EmptySet$&$\DefProp$\\
         \TopVal&$\DefProp$&$\EmptySet$ &$\DefProp$&$\EmptySet$\\
         \tVal&$\DefProp$&$\EmptySet$ &$\EmptySet$&$\DefProp$\\
         \bottomrule
    \end{tabular}\hspace{1cm}%
    \begin{tabular}{>{\centering\arraybackslash}m{.7cm}!{\vrule width .6pt}>{\centering\arraybackslash}m{.7cm}>{\centering\arraybackslash}m{.7cm}>{\centering\arraybackslash}m{.7cm}>{\centering\arraybackslash}m{.7cm}}
        \toprule
         $\DefValueVar$ &$\DiscriminatorVal{\Acc}{\Acc}{\DefValueVar}{\DefBMatrix}$& $\DiscriminatorVal{\Acc}{\NAcc}{\DefValueVar}{\DefBMatrix}$&
         $\DiscriminatorVal{\Rej}{\Rej}{\DefValueVar}{\DefBMatrix}$&
         $\DiscriminatorVal{\Rej}{\NRej}{\DefValueVar}{\DefBMatrix}$\\
         \midrule
         \mcif&$\EmptySet$&$\DefProp$&$\DefProp$&$\EmptySet$\\
         \mciF&$\EmptySet$&$\DefProp$&$\EmptySet$&$\DefProp$\\
         \mciI&$\DefProp,\neg\DefProp$&$\EmptySet$ &$\DefProp$&$\EmptySet$\\
         \mciT&$\DefProp$& $\neg\DefProp$ &$\DefProp$&$\EmptySet$\\
         \mcit&$\DefProp$&$\EmptySet$&$\EmptySet$&$\DefProp$\\
         \bottomrule
    \end{tabular}
\end{table}
\end{example}
The following result 
---which will be instrumental, in particular, within the soundness proof of the axiomatizations that we will develop later on--- 
shows that a discriminator is capable of uniquely characterizing each truth-value of the corresponding
$\theBPN[\Sigma]$--matrix:

\begin{replemma}{lem:caracvalues}
If $\DefBMatrix$ is a
monadic {\upshape$\theBPN[\Sigma]$}--matrix
and $\Discriminator{\DefBMatrix}$
is a discriminator for $\DefBMatrix$,
then, for all $\DefFm \in \DefLangSet$, 
$\DefValueVar \in \DefVSet$
and $\DefBMatrix$-valuation
$\DefVal$,
\begin{equation*}
    \DefVal(\DefFm) = \DefValueVar
    \text{\ \ if{f}\ \ }
    \DefVal(\DiscriminatorVal{\DefCogVar}{\DefCogVar}{\DefValueVar}{\DefBMatrix}(\DefFm))
    \subseteq \BMatrixDistSet{\DefCogVar}{\DefBMatrix}
    \text{ and }
    \DefVal(\DiscriminatorVal{\DefCogVar}{\InvCog{\DefCogVar}}{\DefValueVar}{\DefBMatrix}(\DefFm))
    \subseteq \BMatrixDistSet{{\InvCog{\DefCogVar}}}{\DefBMatrix}
    \text{ for every }
    \DefCogVar \in \Set{\Acc,\Rej}.
\end{equation*}
\end{replemma}
\begin{proof}
Analogous to
the proof of Lemma 1 in \cite{marcelinowollic}.
\end{proof}

\subsection{Calculi for two-dimensional statements}

We may consider the \theB-statements themselves as the
formal objects whose provability by a given (Hilbert-style) deductive proof system we will be interested upon. 
The \theB-statements with finite component sets will be hereupon called
\theB\emph{-sequents}.
A (\emph{\TwoDCalculusTypeName{}) rule schema}
$\DefSchemaVar \SymbDef \TwoDRule{\DefAccSet}{\DefNRejSet}{\DefRejSet}{\DefNAccSet}$ is a \theB-statement 
$\BStat{\DefAccSet}{\DefNRejSet}{\DefRejSet}{\DefNAccSet}$ 
that, when having its component sets subjected
to a substitution $\DefSubs$, produce a \emph{(rule) instance (with schema~$\DefSchemaVar$)}, denoted simply by $\RuleInstance{\DefSchemaVar}{\DefSubs}$;
for each rule instance $\RuleInstance{\DefSchemaVar}{\DefSubs}$,
the pair $\RuleAnt{\SubsApply{\sigma}{\DefAccSet}}{ \SubsApply{\sigma}{\DefRejSet}}$
is said to be the \emph{antecedent}
and the pair $\RuleSuc{\SubsApply{\sigma}{\DefNAccSet}}{\SubsApply{\sigma}{\DefNRejSet}}$ is said to be the \emph{succedent} of $\RuleInstance{\DefSchemaVar}{\DefSubs}$. For later reference, we also set
$\BranchSymb({\RuleInstance{\DefSchemaVar}{\DefSubs}}) \SymbDef \SetSize{\SubsApply{\sigma}{\DefNAccSet} \cup \SubsApply{\sigma}{\DefNRejSet}}$
and $\SeqSizeSymb(\RuleInstance{\DefSchemaVar}{\DefSubs}) \SymbDef \sum_{\DefCogVar}\FmSizeSymb(\SubsApply{\sigma}{\DefFmlaSet_{\DefCogVar}})$, which extends to sets of rule instances in the natural way.
Notice that our notation for rule schemas differs from that of \theB-statements 
with respect to the positioning of the sets of formulas. The purpose is to facilitate the development of proofs in tree form growing downwards from the premises to the conclusion as described in the sequel. \theB-statements, in turn, follow the notation for consequence judgements, which is motivated by the bilattice representation of the four logical values underlying a \theB-consequence relation~\cite{blasiomarcos2017}, in addition to the desire of better expressing the possible interactions between the two dimensions.

A \emph{(\TwoDCalculusTypeName{}) calculus} $\DefCalculusName$ is a collection 
of rule schemas. 
We shall sometimes refer to the set of all rule
instances of a schema $\DefSchemaVar_{}$ of $\DefCalculusName$
as an \emph{inference rule (with schema $\DefSchemaVar_{}$) of $\DefCalculusName$}.
An inference rule with schema
$\DefSchemaVar \SymbDef \TwoDRule{\DefAccSet}{\DefNRejSet}{\DefRejSet}{\DefNAccSet}$
is called \emph{finitary}
whenever $\CogSet{\DefFmlaSet}{\DefCogVar}$
is finite for every $\DefCogVar \in \CogsSet$.
A calculus is finitary
when each of its inference rules
is finitary.

In order to explain what it means for
a \theB-statement 
$\DefBseqVar \SymbDef \BStat{\DefAccSet}{\DefNRejSet}{\DefRejSet}{\DefNAccSet}{}$ to be
provable
--- in other words, for
its succedent
 $\StatSuc{\DefNAccSet}{\DefNRejSet}$
to follow from its antecedent
$\StatAnt{\DefAccSet}{\DefRejSet}$ ---
using the inference rules of a calculus, we will first introduce
the notion of a derivation structured in tree form.
A \emph{directed rooted tree}
$\DefTreeVar$ is a poset
$\Struct{\TreeNodes{\DefTreeVar}, \TreeDescOrder{\DefTreeVar}}$
such that, for every
\emph{node} $\DefNodeVar \in \TreeNodes{\DefTreeVar}$,
the set $\TreeAncests{\DefTreeVar}{\DefNodeVar} \SymbDef \Set{\DefNodeVar^\prime \mid \DefNodeVar^\prime \TreeDescOrderIrr{\DefTreeVar} \DefNodeVar}$ of the \emph{ancestors}
of $\DefNodeVar$ is well-ordered
under $\TreeDescOrderIrr{\DefTreeVar}$,
and there is a single minimal
element $\TreeRoot{\DefTreeVar}$,
called the \emph{root} of
$\DefTreeVar$.
We denote by
$\TreeDescs{\DefNodeVar}{\DefTreeVar} \SymbDef \Set{\DefNodeVar^\prime \mid \DefNodeVar \TreeDescOrderIrr{\DefTreeVar} \DefNodeVar^\prime}$
the set of
\emph{descendants} of~$\DefTreeVar$,
by $\TreeChildren{\DefTreeVar}{\DefNodeVar}$ the minimal elements
of $\TreeDescs{\DefNodeVar}{\DefTreeVar}$ (the \emph{children}
of~$\DefNodeVar$ in~$\DefTreeVar$),
and by $\TreeLeaves{\DefTreeVar}$
the set of maximal elements of
$\TreeDescOrder{\DefTreeVar}$, the \emph{leaves}
of~$\DefTreeVar$.
A rooted tree~$\DefTreeVar$ is
said to be \emph{bounded} when every branch of $\DefTreeVar$ has a leaf.
Moreover, we will call \emph{labelled} 
a rooted tree~$\DefTreeVar$ that comes equipped with a mapping
$\DefTreeLabelling{\DefTreeVar} : \TreeNodes{\DefTreeVar} \to \PowerSet{\DefLangSet}^2 \cup \Set{\StarLabel}$, each node $\DefNodeVar$ of $\DefTreeVar$ being \emph{labelled with} $\DefTreeLabelling{\DefTreeVar}(\DefNodeVar)$.
A node labelled with $\StarLabel$ is said to be \emph{discontinued}.
In what follows, labelled bounded rooted trees will be referred to simply as \emph{trees}.
A tree with a single node labelled with $\mathcal{l} \in \PowerSet{\DefLangSet}^2 \cup \Set{\StarLabel}$ will
be denoted by $\SingleNodeTree(\mathcal{l})$.

Given a node $\DefNodeVar$
labelled with $\Tuple{\DefFmlaSet, \DeffFmlaSet}$
and given a formula $\DefFm$,
we shall use $\ExpandNodeAcc{\DefNodeVar}{\DefFm}$
to refer to a node labelled with
$\Tuple{\DefFmlaSet \cup \Set{\DefFm}, \DeffFmlaSet}$
and use $\ExpandNodeRej{\DefNodeVar}{\DefFm}$
to refer to a node labelled with
$\Tuple{\DefFmlaSet, \DeffFmlaSet \cup \Set{\DefFm}}$.
We say that a tree $\DefTreeVar$
is a $\DefCalculusName$\emph{-derivation} provided that
for each non-leaf node~$\DefNodeVar$ of~$\DefTreeVar$
labelled with 
$\Tuple{\AccSet{\DeffFmlaSet}, \RejSet{\DeffFmlaSet}}$
there is
an instance of an inference rule of $\DefCalculusName$,
say
$\RuleInstance{\DefSchemaVar}{\DefSubs} = \TwoDRule
{\SubsApply{\DefSubs}{\DefAccSet}}
{\SubsApply{\DefSubs}{\DefNRejSet}}
{\SubsApply{\DefSubs}{\DefRejSet}}
{\SubsApply{\DefSubs}{\DefNAccSet}}{}$, that \textit{expands} $\DefNodeVar$ or, equivalently, that is \textit{applicable} to the label of $\DefNodeVar$,
meaning that $\SubsApply{\DefSubs}{\CogSet{\DefFmlaSet}{\DefCogVar}} \subseteq \CogSet{\DeffFmlaSet}{\DefCogVar}$,
for every $\DefCogVar \in \Set{\Acc, \Rej}$,
and
\begin{itemize}
    \item 
    if $\DefNAccSet \cup \DefNRejSet = \EmptySet$,
    then $\TreeChildren{\DefTreeVar}{\DefNodeVar} = \Set{\StarTree{\DefNodeVar}}$ and $\DefTreeLabelling{\DefTreeVar}(\StarTree{\DefNodeVar}) = \StarLabel$
    
    \item 
    otherwise,
    $\TreeChildren{\DefTreeVar}{\DefNodeVar} = \Set{\ExpandNodeAcc{\DefNodeVar}{\DefFm} \mid \DefFm \in \SubsApply{\DefSubs}{\NAccSet{\DefFmlaSet}}} \cup \Set{\ExpandNodeRej{\DefNodeVar}{\DefFm} \mid \DefFm \in \SubsApply{\DefSubs}{\NRejSet{\DefFmlaSet}}}$
\end{itemize}

\noindent 
We should observe that, with our present notation, traditional Hilbert-style
derivations (when only inference
rules with a single formula
in the succedent are applied)
turn out to be linear trees; 
for all practical purposes, at any given node we may count with all the information from previous nodes in the branch, and, accordingly, a rule application with a single succedent just adds a new bit of information to that very branch.

Given 
a \theB-statement
$\DefBseqVar \SymbDef \BStat{\DefAccSet}{\DefNRejSet}{\DefRejSet}{\DefNAccSet}$
and a calculus $\DefCalculusName$,
a $\DefCalculusName$-derivation
$\DefTreeVar$
with 
$\DefTreeLabelling{\DefTreeVar}(\TreeRoot{\DefTreeVar}) = \Tuple{\AccSet{\DeffFmlaSet}, \RejSet{\DeffFmlaSet}}$
is a
$\DefCalculusName$\emph{-proof}
of~$\DefBseqVar$
provided that
$\CogSet{\DeffFmlaSet}{\DefCogVar} \subseteq \CogSet{\DefFmlaSet}{\DefCogVar}$
for every $\DefCogVar \in \Set{\Acc, \Rej}$
and, for all $\DefNodeVar \in \TreeLeaves{\DefTreeVar}$
with $\DefTreeLabelling{\DefTreeVar}(\DefNodeVar) = \Tuple{\NAccSet{\DeffFmlaSet}, \NRejSet{\DeffFmlaSet}}$,
we have
$\CogSet{\DeffFmlaSet}{\DefCogVar} \cap \CogSet{\DefFmlaSet}{\DefCogVar} \neq \EmptySet$
for some $\DefCogVar \in \Set{\NAcc, \NRej}$.
We also say that 
a node
is \emph{$\Tuple{\DefNAccSet,\DefNRejSet}$-closed}
when the latter condition holds for such node and we say that
$\DefTreeVar$ is
\emph{$\Tuple{\DefNAccSet,\DefNRejSet}$-closed} when all
of its leaf nodes are $\Tuple{\DefNAccSet,\DefNRejSet}$-closed.
When a $\DefCalculusName$-proof
exists for the \theB-statement
$\DefBseqVar$, we say that
$\DefBseqVar$ is
$\DefCalculusName$\emph{-provable}.
The reader is referred to Example~\ref{ex:derivations}
in order to see some proofs of the form we have just described.
A calculus
$\DefCalculusName$ induces
a $2{\times}2$-place relation $\BCalcConName{\DefCalculusName}$
over $\PowerSet{\DefLangSet}$
such that $\BCalcCon{\DefAccSet}{\DefNRejSet}{\DefRejSet}{\DefNAccSet}{\DefCalculusName}$
whenever $\BStat{\DefAccSet}{\DefNRejSet}{\DefRejSet}{\DefNAccSet}$ is $\DefCalculusName$-provable.
As we point out in Prop.~\ref{prop:calculusbconseq} below, this provides another realization (compare with Prop.\ \ref{propo:bentailmentbconseq}) of a \theB-consequence relation.

\begin{proposition}
\label{prop:calculusbconseq}
Given a calculus $\DefCalculusName$,
the $2{\times}2$-place relation $\BCalcConName{\DefCalculusName}$
is the smallest \theB-consequence
containing the rules of
$\DefCalculusName$.
\end{proposition}

Given a collection~$R$ of rule instances, 
we say that a \theB-statement $\DefBseqVar$
is $R$-provable whenever there is a proof
of~$\DefBseqVar$ using only rule instances
in~$R$. We may define a  $2{\times}2$-place relation
$\BRuleInstCon{\PlaceholderArg}{\PlaceholderArg}{\PlaceholderArg}{\PlaceholderArg}{R}$ by setting
$\BRuleInstCon{\DefAccSet}{\DefNRejSet}{\DefRejSet}{\DefNAccSet}{R}$ to hold iff $\BStat{\DefAccSet}{\DefNRejSet}{\DefRejSet}{\DefNAccSet}{}$ is $R$-provable. Although not necessarily substitution-invariant, one may readily check that
this relation respects properties \ref{prop:BConO}, \ref{prop:BConD}
and \ref{prop:BConC}.

Given a $\theBPN[\Sigma]$--matrix~$\DefBMatrix$,
we say that a calculus $\DefCalculusName$
is \emph{sound} with respect to~$\DefBMatrix$
whenever $\BCalcConName{\DefCalculusName} \subseteq \BEnt{\PlaceholderArg}{\PlaceholderArg}{\PlaceholderArg}{\PlaceholderArg}{\DefBMatrix}$
and say that it is \emph{complete} with respect to~$\DefBMatrix$ when the converse inclusion holds.
Being sound and complete means that
$\DefCalculusName$
\emph{axiomatizes}~$\DefBMatrix$.

\begin{example}
Any fully indeterministic $\theBPN[\Sigma]$--matrix
is axiomatized by the empty set of rules.
\end{example}

\begin{example}
\label{ex:axiomatinfo}
We present below a calculus that axiomatizes
the $\theBPN[\Sigma]$--matrix $\InfoFourMatrix$
introduced in Example \ref{ex:infomatrix}, resulting
from the simplification
of the calculus produced via the recipe described
in Definition~\ref{def:axiomatization},
given further ahead.
\[
    \TwoDRule{\DefProp}{}{}{\DefProp \lor \DeffProp}{\lor^4_1}\quad
    \TwoDRule{\DeffProp}{}{}{\DefProp \lor \DeffProp}{\lor^4_2}\quad
    \TwoDRule{}{\DefProp \lor \DeffProp}{\DefProp, \DeffProp}{}{\lor^4_3}\quad
    \TwoDRule{}{\DeffProp}{\DefProp \lor \DeffProp}{}{\lor^4_4}\quad
    \TwoDRule{}{\DefProp}{\DefProp \lor \DeffProp}{}{\lor^4_5}
\]
\[
    \TwoDRule{\DefProp\land\DeffProp}{}{}{\DefProp}{\land^4_1}\quad
    \TwoDRule{\DefProp\land\DeffProp}{}{}{\DeffProp}{\land^4_2}\quad
    \TwoDRule{\DefProp, \DeffProp}{}{}{\DefProp\land\DeffProp}{\land^4_3}\quad
    \TwoDRule{}{\DefProp \land \DeffProp}{\DeffProp}{}{\land^4_4}\quad
    \TwoDRule{}{\DefProp \land \DeffProp}{\DefProp}{}{\land^4_5}\\
\]
\[
    \TwoDRule{}{}{\neg\DefProp}{\DefProp}{\neg^4_1}\quad
    \TwoDRule{}{}{\DefProp}{\neg \DefProp}{\neg^4_2}\quad
    \TwoDRule{\neg\DefProp}{\DefProp}{}{}{\neg^4_3}\quad
    \TwoDRule{\DefProp}{\neg \DefProp}{}{}{\neg^4_4}
\]
\end{example}

The next example illustrates how
adding rules to an axiomatization
of a $\theBPN[\Sigma]$--matrix
$\DefBMatrix$
imposes
refinements on $\DefBMatrix$
in order to
guarantee soundness of
these very rules.
Such mechanism is essential
to the axiomatization
procedure presented in the
next section.

\begin{example}
\label{ex:otheraxiomats}
We obtain an axiomatization
for $\FDEMatrix$ by adding
the following rules to the
calculus of Example \ref{ex:axiomatinfo}:
\[
    \TwoDRule{\DefProp\lor\DeffProp}{}{}{\DefProp,\DeffProp}{\lor^4_6}
    \qquad
    \TwoDRule{}{\DefProp, \DeffProp}{\DefProp\land\DeffProp}{}{\land^4_6} 
\]
If, in addition, we include
the rule
\[
\TwoDRule{\DeffProp}{\DefProp}{\DeffProp}{\DefProp}{\mathrm{T}^4}
\]
we axiomatize $\ImpFreeKleeneMatrix$ (see Example~\ref{ex:kleenemat}).

Let us explain the intuition behind this mechanism considering the case of rule~$\land^4_6$; the other rules will follow the same principle. 
What rule $\land^4_6$ enforces is that any refinement of
$\InfoFourMatrix$ with respect
to which this rule is
sound must disallow
valuations that assign
values in $\Set{\BotVal,\tVal}$
to formulas $\DefFm$ and $\DeffFm$
while assigning 
a value in $\Set{\TopVal,\fVal}$
to $\DefFm\land\DeffFm$,
for otherwise such valuation would
constitute a countermodel for that very rule.
This is reflected in $\BMatrixInterp{\land}{\FDEMatrix}$
(Example \ref{ex:fde}) by the absence of the values from the set $\Set{\TopVal,\fVal}$
in the entries corresponding to the truth-value assignments in which both inputs belong to $\Set{\BotVal,\tVal}$.
\end{example}

\begin{example}
By the same mechanism used in the previous example, in adding the rules
$\TwoDRule{}{\DefProp}{}{\DefProp}{\mathrm{\BotVal E}}$ and
$\TwoDRule{\DefProp}{}{\DefProp}{}{\mathrm{\TopVal E}}$
to the axiomatization of $\FDEMatrix$, we force empty
outputs on any truth-table entry whose input involves 
either $\BotVal$ or $\TopVal$. 
It follows that Classical Logic inhabits the
$\tAsp$-aspect of the resulting $\theBPN[\Sigma]$--matrix, hereby called~$\CLMatrix$.
\end{example}

\begin{example}
\label{ex:derivations}
In Figure \ref{fig:derivs}, we offer proofs of 
$\BStat{\neg(\DefProp \land \DeffProp)}{}{}{\neg\DefProp \lor \neg\DeffProp}{}$, 
$\BStat{\DefffProp\land\DefProp}{\DeffffProp}{\DefProp\lor\DeffProp}{\DeffffProp}$
and $\BStat{\DefProp,\neg\DefProp}{}{}{}$, respectively,
in the calculi for
$\FDEMatrix$, $\ImpFreeKleeneMatrix$
and $\CLMatrix$
presented in the previous examples.
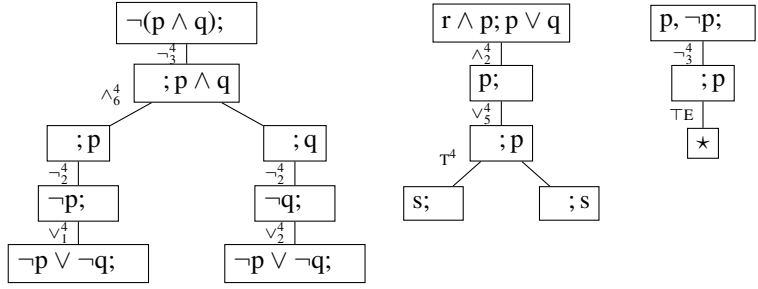
\begin{figure}[thb]
    \centering
    \begin{tikzpicture}[every tree node/.style={draw},
       level distance=0.8cm,sibling distance=1cm,
       edge from parent path={(\tikzparentnode) -- (\tikzchildnode)}, baseline]
    \Tree[.\DNode{\neg(\DefProp\land\DeffProp)}{}
            \edge node[auto=right] {\DRule{\neg^4_3}};
            [.\DNode{}{\DefProp\land\DeffProp}
                \edge node[auto=right] {\DRule{\land^4_6}};
                [.\DNode{}{\DefProp}
                    \edge node[auto=right] {\DRule{\neg^4_2}};
                    [.\DNode{\neg\DefProp}{}
                        \edge node[auto=right] {\DRule{\lor^4_1}};
                        [.\DNode{\neg\DefProp\lor\neg\DeffProp}{}
                        ]
                    ]
                ]
                [.\DNode{}{\DeffProp}
                    \edge node[auto=right] {\DRule{\neg^4_2}};
                    [.\DNode{\neg\DeffProp}{}
                        \edge node[auto=right] {\DRule{\lor^4_2}};
                        [.\DNode{\neg\DefProp\lor\neg\DeffProp}{}
                        ]
                    ]
                ]
            ]
    ]
    \end{tikzpicture}%
    \hspace{.5cm}%
    \begin{tikzpicture}[every tree node/.style={draw},
       level distance=0.8cm,sibling distance=1cm,
       edge from parent path={(\tikzparentnode) -- (\tikzchildnode)}, baseline]
    \Tree[.\DNode{\DefffProp\land\DefProp}{\DefProp\lor\DeffProp}
        \edge node[auto=right] {\DRule{\land^4_2}};
        [.\DNode{\DefProp}{} 
            \edge node[auto=right] {\DRule{\lor^4_5}};
            [.\DNode{}{\DefProp}
                \edge node[auto=right] {\DRule{\mathrm{T}^4}};
                [.\DNode{\DeffffProp}{}
                ]
                [.\DNode{}{\DeffffProp}
                ]
            ]
        ]
    ]
    \end{tikzpicture}
    \hspace{.5cm}
    \begin{tikzpicture}[every tree node/.style={draw},
       level distance=0.8cm,sibling distance=1cm,
       edge from parent path={(\tikzparentnode) -- (\tikzchildnode)}, baseline]
    \Tree[.\DNode{\DefProp,\neg\DefProp}{}
        \edge node[auto=right] {\DRule{\neg^4_3}};
        [.\DNode{}{\DefProp} 
            \edge node[auto=right] {\DRule{\TopVal E}};
            [.$\star$
            ]
        ]
    ]
    \end{tikzpicture}
\setlength{\belowcaptionskip}{-15pt}
    \caption{Examples of derivations in tree form. For the sake of a cleaner presentation, we omit the formulas that are inherited when expanding a node.
    }
    \label{fig:derivs}
\end{figure}
\end{example}

We conclude this section by introducing the notion of (generalized) analyticity of a calculus, an important feature for proof-search procedures that is built in the axiomatizations delivered by the recipe of the next section. Given a \theB-state\-ment $\DefBseqVar \SymbDef \BStat{\DefAccSet}{\DefNRejSet}{\DefRejSet}{\DefNAccSet}$, let
$\SubfmlasSeq(\DefBseqVar) \SymbDef \bigcup_{\DefCogVar \in \CogsSet} \Subfmlas{\CogSet{\DefFmlaSet}{\DefCogVar}}$
be the collection of \emph{subformulas of
$\DefBseqVar$}, and
$\GenSubfmlas{}{\DeffFmlaSet}(\DefBseqVar) \SymbDef 
\SubfmlasSeq(\DefBseqVar) \cup \Set{\SubsApply{\sigma}{\DefFm} \mid \DefFm \in \DeffFmlaSet, \DefSubs : \DefProps \to \SubfmlasSeq(\DefBseqVar)}
$ be
the \emph{generalized subformulas of $\DefBseqVar$ (with respect to $\DeffFmlaSet$)}.
Define the $2{\times}2$-place relation 
$\BAnaCalcCon{\PlaceholderArg}{\PlaceholderArg}{\PlaceholderArg}{\PlaceholderArg}{\DefCalculusName}{\GenSubfmlas{}{\DeffFmlaSet}}$
over
$\PowerSet{\DefLangSet}$ by setting
$\BAnaCalcCon{\DefAccSet}{\DefNRejSet}{\DefRejSet}{\DefNAccSet}{\DefCalculusName}{\GenSubfmlas{}{\DeffFmlaSet}}$
if{f} 
there is a $\DefCalculusName$-proof $\DefTreeVar$ of
$\DefBseqVar \SymbDef \BStat{\DefAccSet}{\DefNRejSet}{\DefRejSet}{\DefNAccSet}$ such that
$\DefTreeLabelling{\DefTreeVar}(\TreeNodes{\DefTreeVar}) \subseteq \PowerSet{\GenSubfmlas{}{\DeffFmlaSet}(\DefBseqVar)}^2 \cup \Set{\StarLabel}$.
Such a proof is said to be \emph{$\DeffFmlaSet$-an\-alyt\-ic}.
We say that
$\DefCalculusName$ is
\emph{$\DeffFmlaSet$-an\-alyt\-ic}
in case
$\BCalcCon{\DefAccSet}{\DefNRejSet}{\DefRejSet}{\DefNAccSet}{\DefCalculusName}$ implies
$\BAnaCalcCon{\DefAccSet}{\DefNRejSet}{\DefRejSet}{\DefNAccSet}{\DefCalculusName}{\GenSubfmlas{}{\DeffFmlaSet}}$.
We will denote by
$\RuleInstSetCalcSeq{\DefCalculusName}{\DefBseqVar}$
the set of all rule instances of $\DefCalculusName$ resulting from substitutions
that only use formulas in $\GenSubfmlas{}{\DeffFmlaSet}(\DefBseqVar)$.

%% file: tex/axiomatizing.tex
\section{Axiomatizing Monadic $\theBPN[\Sigma]$-matrices}
\label{sec:axiomatization}


We now describe four collections of rule schemas by which any
sufficiently expressive $\theBPN[\Sigma]$--matrix
$\DefBMatrix$ is constrained. 
Together, these schemas constitute a presentation of a calculus that will be denoted by $\BIndAxiomatName{\DefBMatrix}{\Discriminator{\DefBMatrix}}{}$, where $\Discriminator{\DefBMatrix}$ is a discriminator for~$\DefBMatrix$.
The first collection, $\BIndAxiomatRuleExistsName$, is intended to exclude all combinations of separators that do not correspond to truth-values.
The second, $\BIndAxiomatRuleDName$, sets
the combinations of separators that characterize acceptance 
apart from those that characterize  non-acceptance,
and sets the combinations of separators that characterize rejection
apart from those that characterize  non-rejection.
The third one, 
$\BIndAxiomatRuleSigmaName$, 
fully describes, through appropriate refinements, the interpretation of the connectives of~$\DefSig$ in $\DefBMatrix$. 
At last, the rules in
$\BIndAxiomatRuleNonTotalName$
guarantee that values belong to total sub--$\theBPN[\DefSig]$--matrices of~$\DefBMatrix$.%

In what follows, given $\DefSetValueVar \subseteq \DefVSet$, we shall use
$\Tuple{\DiscriminatorValCh{\Acc}{\Acc}{\DefSetValueVar}{\DefBMatrix},\DiscriminatorValCh{\Rej}{\Rej}{\DefSetValueVar}{\DefBMatrix}}$
to denote a pair of sets in which
$\DiscriminatorValCh{\DefCogVar}{\DefCogVar}{\DefSetValueVar}{\DefBMatrix}$, with
$\DefCogVar \in \Set{\Acc,\Rej}$,
is obtained by choosing
an element of
$\DiscriminatorVal{\DefCogVar}{\DefCogVar}{\DefValueVar}{\DefBMatrix}$
for each $\DefValueVar \in \DefSetValueVar$.
Notice that, when $\DefSetValueVar = \EmptySet$,
the only possibility is the pair $\Tuple{\EmptySet,\EmptySet}$; moreover, 
when $\DiscriminatorVal{\Acc}{\Acc}{\DefValueVar}{\DefBMatrix} \cup \DiscriminatorVal{\Rej}{\Rej}{\DefValueVar}{\DefBMatrix} = \EmptySet$ for some $\DefValueVar \in \DefSetValueVar$, no such pair exists.
The pair
$\Tuple{\DiscriminatorValCh{\Acc}{\NAcc}{\DefSetValueVar}{\DefBMatrix},\DiscriminatorValCh{\Rej}{\NRej}{\DefSetValueVar}{\DefBMatrix}}$ shall be used analogously.
\begin{definition}
\label{def:axiomatization}
Let $\DefBMatrix$ be
a \upshape{$\theBPN[\Sigma]$}--matrix, and let $\Discriminator{\DefBMatrix}$ be a discriminator for 
$\DefBMatrix$.
The calculus
$\BIndAxiomatName{\DefBMatrix}{\Discriminator{\DefBMatrix}}{}$
is presented by way of the following rule schemas:
\begin{namedproperties}\itemsep0pt
    \item[($\BIndAxiomatRuleExistsName$)\label{calc:exists}] 
    \sloppy
    for each $\DefSetValueVarProp \subseteq \DefVSet$
    and each possible choices of
    $\Tuple{\DiscriminatorValCh{\Acc}{\Acc}{\VSetComp{\DefSetValueVar}}{\DefBMatrix},\DiscriminatorValCh{\Rej}{\Rej}{\VSetComp{\DefSetValueVar}}{\DefBMatrix}}$ and of $\Tuple{\DiscriminatorValCh{\Acc}{\NAcc}{\DefSetValueVarProp}{\DefBMatrix},\DiscriminatorValCh{\Rej}{\NRej}{\DefSetValueVarProp}{\DefBMatrix}}$,    with $\VSetComp{\DefSetValueVar} \SymbDef \SetDiff{\DefVSet}{\DefSetValueVarProp}$,
    \[
        \TwoDRule
        {\DiscriminatorValCh{\Acc}{\NAcc}{\DefSetValueVarProp}{\DefBMatrix}}
        {\DiscriminatorValCh{\Rej}{\Rej}{\VSetComp{\DefSetValueVar}}{\DefBMatrix}}
        {\DiscriminatorValCh{\Rej}{\NRej}{\DefSetValueVarProp}{\DefBMatrix}}
        {\DiscriminatorValCh{\Acc}{\Acc}{\VSetComp{\DefSetValueVar}}{\DefBMatrix}}
    \]

    \item[($\BIndAxiomatRuleDName$)\label{calc:d}] 
    for an arbitrary propositional variable $\DefProp \in \DefProps$, and for each $\DefValueVar \in \DefVSet$, 
    \[
        \TwoDRule
        {\DiscriminatorVal{\Acc}{\Acc}{\DefValueVar}{\DefBMatrix}(\DefProp), \FmlaPlacerByValue{\NAcc}(\DefValueVar)}
        {\DiscriminatorVal{\Rej}{\NRej}{\DefValueVar}{\DefBMatrix}(\DefProp)}
        {\DiscriminatorVal{\Rej}{\Rej}{\DefValueVar}{\DefBMatrix}(\DefProp)}
        {\DiscriminatorVal{\Acc}{\NAcc}{\DefValueVar}{\DefBMatrix}(\DefProp), \FmlaPlacerByValue{\Acc}(\DefValueVar)}
        \qquad
        \TwoDRule
        {\DiscriminatorVal{\Acc}{\Acc}{\DefValueVar}{\DefBMatrix}(\DefProp)}
        {\DiscriminatorVal{\Rej}{\NRej}{\DefValueVar}{\DefBMatrix}(\DefProp), \FmlaPlacerByValue{\Rej}(\DefValueVar)}
        {\DiscriminatorVal{\Rej}{\Rej}{\DefValueVar}{\DefBMatrix}(\DefProp),  \FmlaPlacerByValue{\NRej}(\DefValueVar)}
        {\DiscriminatorVal{\Acc}{\NAcc}{\DefValueVar}{\DefBMatrix}(\DefProp)}
    \]
    where, for $\DefCogVar \in \CogsSet$, $\FmlaPlacerByValue{\DefCogVar} : \DefVSet \to \PowerSet{\Set{\DefProp}}$
    is such that
    $\FmlaPlacerByValue{\DefCogVar}(\DefValueVar) = \Set{\DefProp}$
    if{f} $\DefValueVar \in \BMatrixDistSet{\DefCogVar}{\DefBMatrix}$.

    \item[($\BIndAxiomatRuleSigmaName$)\label{calc:sigma}] for each $k$-ary connective $\DefSymbol$, each sequence 
    $X \SymbDef (\DefValueVar_1,\ldots,\DefValueVar_k)$
    of
    truth-values of $\DefBMatrix$, each $\DeffValueVar \not\in 
    \BMatrixInterp{\DefSymbol}{\DefBMatrix}X
    $,
    and for a sequence
    $(\DefProp_1,\ldots,\DefProp_k)$
    of distinct propositional variables,
    \[
        \TwoDRule
        {\SigmaRulesSet{\Acc}{\DefSymbol, X,\DeffValueVar}}
        {\SigmaRulesSet{\NRej}{\DefSymbol, X,\DeffValueVar}}
        {\SigmaRulesSet{\Rej}{\DefSymbol, X,\DeffValueVar}}
        {\SigmaRulesSet{\NAcc}{\DefSymbol, X,\DeffValueVar}}
    \]
    where each $\SigmaRulesSet{\DefCogVar}{\DefSymbol, X,\DeffValueVar} \SymbDef \bigcup\limits_{1 \leq i \leq k}
            {\DiscriminatorVal{\DefCogVar}{\DefCogVar}{\DefValueVar_i}{\DefBMatrix}(\DefProp_i)}\cup
            {\DiscriminatorVal{\DefCogVar}{\DefCogVar}{\DeffValueVar}{\DefBMatrix}(\DefSymbol(\DefProp_1,\ldots,\DefProp_k))}$.

    \item[($\BIndAxiomatRuleNonTotalName$)\label{calc:nontotal}]
        for each 
    $\DefSetValueVar \not\in \TotalSubMVSet{\DefBMatrix}$
    and an arbitrary family
    $\Family{\DefProp_\DefValueVar}{\DefValueVar \in \DefSetValueVar}$
    of distinct propositional variables,
    \[
        \TwoDRule
        {\bigcup_{\DefValueVar \in \DefSetValueVar}\DiscriminatorVal{\Acc}{\Acc}{\DefValueVar}{\DefBMatrix}(\DefProp_\DefValueVar)}
        {\bigcup_{\DefValueVar \in \DefSetValueVar}\DiscriminatorVal{\Rej}{\NRej}{\DefValueVar}{\DefBMatrix}(\DefProp_\DefValueVar)}
        {\bigcup_{\DefValueVar \in \DefSetValueVar}\DiscriminatorVal{\Rej}{\Rej}{\DefValueVar}{\DefBMatrix}(\DefProp_\DefValueVar)}
        {\bigcup_{\DefValueVar \in \DefSetValueVar}\DiscriminatorVal{\Acc}{\NAcc}{\DefValueVar}{\DefBMatrix}(\DefProp_\DefValueVar)}.
    \] 
\end{namedproperties}
\end{definition}

\begin{reptheorem}{the:soundness}
If $\Discriminator{\DefBMatrix}$ is a discriminator for
a \upshape{$\theBPN[\Sigma]$}--matrix
$\DefBMatrix$, then
the calculus
$\BIndAxiomatName{\DefBMatrix}{\Discriminator{\DefBMatrix}}{}$
is sound with respect to
$\DefBMatrix$.
\end{reptheorem}
\begin{proof}
We can show by contradiction that no
$\DefBMatrix$-valuation
can be a countermodel for the schemas in each of the groups of schemas of $\BIndAxiomatName{\DefBMatrix}{\Discriminator{\DefBMatrix}}{}$. 
%
We detail the case of \ref{calc:exists}.
Consider a schema
    $\DefBseqVar \SymbDef \TwoDRuleSm
        {\DiscriminatorValCh{\Acc}{\NAcc}{\DefSetValueVarProp}{\DefBMatrix}}
        {\DiscriminatorValCh{\Rej}{\Rej}{\VSetComp{\DefSetValueVar}}{\DefBMatrix}}
        {\DiscriminatorValCh{\Rej}{\NRej}{\DefSetValueVarProp}{\DefBMatrix}}
        {\DiscriminatorValCh{\Acc}{\Acc}{\VSetComp{\DefSetValueVar}}{\DefBMatrix}}$,
        for some $\DefSetValueVarProp \subseteq \DefVSet$
        and some choice of
         $\Tuple{\DiscriminatorValCh{\Acc}{\Acc}{\VSetComp\DefSetValueVar}{\DefBMatrix},\DiscriminatorValCh{\Rej}{\Rej}{\VSetComp\DefSetValueVar}{\DefBMatrix}}$ and $\Tuple{\DiscriminatorValCh{\Acc}{\NAcc}{\DefSetValueVarProp}{\DefBMatrix},\DiscriminatorValCh{\Rej}{\NRej}{\DefSetValueVarProp}{\DefBMatrix}}$.
         Suppose that 
         $\DefBseqVar$
         does not hold in $\DefBMatrix$, 
         with the valuation $\DefVal$ witnessing this fact. We will prove that,
         given a propositional variable
         $\DefProp$, $\DefVal(\DefProp) \neq \DefValueVar$, for all $\DefValueVar \in \DefVSet$, an absurd.
         For that purpose, let $\DefValueVar \in \DefVSet$.
         In case $\DefValueVar \in \DefSetValueVarProp$,
         there must be a separator
         $\DefSeparator$ in
         $\DiscriminatorVal{\DefCogVar}{\InvCog{\DefCogVar}}{\DefValueVar}{\DefBMatrix}$, for some $\DefCogVar \in \Set{\Acc,\Rej}$, such that
         $\DefVal(\DefSeparator(\DefProp)) \in \BMatrixDistSet{\DefCogVar}{\DefBMatrix}$.
         By Lemma \ref{lem:caracvalues}, this implies 
         that $\DefVal(\DefProp) \neq \DefValueVar$.
         The reasoning is similar in case $\DefValueVar \in \VSetComp{\DefSetValueVar}$.
\end{proof}

In what follows, 
denote by~$\GenSubfmlas{}{\Discriminator{\DefBMatrix}}$
the mapping $\GenSubfmlas{}{\SeparatorsSet{\DefBMatrix}}$,
which indicates what formulas may appear in a 
$\SeparatorsSet{\DefBMatrix}$-analytic proof.
In order to prove completeness and
$\SeparatorsSet{\DefBMatrix}$-analyticity of 
$\BIndAxiomatName{\DefBMatrix}{\Discriminator{\DefBMatrix}}{}$
with respect to $\DefBMatrix$,
we shall make use of Lemma~\ref{lem:auxcompleteness} presented below,
which contains four items, each one referring
to a group of schemas of $\BIndAxiomatName{\DefBMatrix}{\Discriminator{\DefBMatrix}}{}$. Intuitively, 
given a \theB-statement $\DefBseqVar$ 
and assuming that
there is no $\SeparatorsSet{\DefBMatrix}$-analytic
proof of it in $\BIndAxiomatName{\DefBMatrix}{\Discriminator{\DefBMatrix}}{}$,
items 1 and 2 give us the resources
to define a mapping 
$f : \Subfmlas{\DefBseqVar} \to \DefVSet$
that, by items 3 and 4, can be extended to 
a countermodel for $\DefBseqVar$
in $\DefBMatrix$.

\begin{replemma}{lem:auxcompleteness}
For all \theB-statements $\DefBseqVar$ of the form $ \BStat{\IntermAccSet{\CutPropSet}}{\SetCompl{\IntermRejSet{\CutPropSet}}}{\IntermRejSet{\CutPropSet}}{\SetCompl{\IntermAccSet{\CutPropSet}}}${\upshape :}
\begin{enumerate}
    \item if $\nBAnaCalcCon{\IntermAccSet{\CutPropSet}}{\SetCompl{\IntermRejSet{\CutPropSet}}}{\IntermRejSet{\CutPropSet}}{\SetCompl{\IntermAccSet{\CutPropSet}}}{\BIndAxiomatRuleExistsName}{\GenSubfmlas{}{\Discriminator{\DefBMatrix}}}$,
    then for all $\DefFm \in \Subfmlas{\DefBseqVar}$
    there is an $\DefValueVar \in \DefVSet$ such that
    $\DiscriminatorVal{\DefCogVar}{\DefCogVar}{\DefValueVar}{\DefBMatrix}(\DefFm) \subseteq \CutPropSet_\DeffCogVar$
    and $\DiscriminatorVal{\DefCogVar}{\InvCog{\DefCogVar}}{\DefValueVar}{\DefBMatrix}(\DefFm) \subseteq \SetCompl{\CutPropSet_\DeffCogVar}$, for
    $\Tuple{\DefCogVar,\DeffCogVar} \in \Set{\Tuple{\Acc,\TheV}, \Tuple{\Rej,\TheNV}}$;
    \item if $\nBAnaCalcCon{\IntermAccSet{\CutPropSet}}{\SetCompl{\IntermRejSet{\CutPropSet}}}{\IntermRejSet{\CutPropSet}}{\SetCompl{\IntermAccSet{\CutPropSet}}}{\BIndAxiomatRuleDName}{\GenSubfmlas{}{\Discriminator{\DefBMatrix}}}$,
    then for every $\DefFm \in \Subfmlas{\DefBseqVar}$
    and $\DefValueVar \in \DefVSet$ such that
    $\DiscriminatorVal{\DefCogVar}{\DefCogVar}{\DefValueVar}{\DefBMatrix}(\DefFm) \subseteq \CutPropSet_\DeffCogVar$
    and $\DiscriminatorVal{\DefCogVar}{\InvCog{\DefCogVar}}{\DefValueVar}{\DefBMatrix}(\DefFm) \subseteq \SetCompl{\CutPropSet_\DeffCogVar}$, we have
    $\DefValueVar \in \BMatrixDistSet{\DefCogVar}{\DefBMatrix}$
    if{f} $\DefFm \in \CutPropSet_\DeffCogVar$, for
    $\Tuple{\DefCogVar,\DeffCogVar} \in \Set{\Tuple{\Acc,\TheV}, \Tuple{\Rej,\TheNV}}$;
    \item if $\nBAnaCalcCon{\IntermAccSet{\CutPropSet}}{\SetCompl{\IntermRejSet{\CutPropSet}}}{\IntermRejSet{\CutPropSet}}{\SetCompl{\IntermAccSet{\CutPropSet}}}{\BIndAxiomatRuleSigmaName}{\GenSubfmlas{}{\Discriminator{\DefBMatrix}}}$,
    then for every $\DefSymbol \in \SigAritySet{\Sigma}{k}$,
    $\DefFm \SymbDef \DefSymbol(\DefFm_1,\ldots,\DefFm_k) \in \Subfmlas{\DefBseqVar}$
    and $\DefValueVar_1,\ldots,\DefValueVar_k \in \DefVSet$
    with
    $\DiscriminatorVal{\DefCogVar}{\DefCogVar}{\DefValueVar_i}{\DefBMatrix}(\DefFm_i) \subseteq \CutPropSet_\DeffCogVar$
    and $\DiscriminatorVal{\DefCogVar}{\InvCog{\DefCogVar}}{\DefValueVar_i}{\DefBMatrix}(\DefFm_i) \subseteq \SetCompl{\CutPropSet_\DeffCogVar}$, for
    each $1 \leq i \leq k$ and $\Tuple{\DefCogVar,\DeffCogVar} \in \Set{\Tuple{\Acc,\TheV}, \Tuple{\Rej,\TheNV}}$, we have that
    $\DiscriminatorVal{\DefCogVar}{\DefCogVar}{\DeffValueVar}{\DefBMatrix}(\DefFm) \subseteq \CutPropSet_\DeffCogVar$
    and $\DiscriminatorVal{\DefCogVar}{\InvCog{\DefCogVar}}{\DeffValueVar}{\DefBMatrix}(\DefFm) \subseteq \SetCompl{\CutPropSet_\DeffCogVar}$ for each $\Tuple{\DefCogVar,\DeffCogVar} \in \Set{\Tuple{\Acc,\TheV}, \Tuple{\Rej,\TheNV}}$ implies
    $\DeffValueVar \in \BMatrixInterp{\DefSymbol}{\DefBMatrix}(\DefValueVar_1,\ldots,\DefValueVar_k)$;
    \item if $\nBAnaCalcCon{\IntermAccSet{\CutPropSet}}{\SetCompl{\IntermRejSet{\CutPropSet}}}{\IntermRejSet{\CutPropSet}}{\SetCompl{\IntermAccSet{\CutPropSet}}}{\BIndAxiomatRuleNonTotalName}{\GenSubfmlas{}{\Discriminator{\DefBMatrix}}}$,
    then
    $\Set{\DefValueVar \in \DefVSet \mid  \DiscriminatorVal{\DefCogVar}{\DefCogVar}{\DefValueVar}{\DefBMatrix}(\DefFm) \subseteq \CutPropSet_\DeffCogVar
    \text{ and } \DiscriminatorVal{\DefCogVar}{\InvCog{\DefCogVar}}{\DefValueVar}{\DefBMatrix}(\DefFm) \subseteq \SetCompl{\CutPropSet_\DeffCogVar},\text{ }\\
    \text{\qquad\qquad\qquad\qquad\;\;\,} \text{for each }\Tuple{\DefCogVar,\DeffCogVar} \!\in\! \Set{\Tuple{\Acc,\TheV}, \Tuple{\Rej,\TheNV}}
    \text{ and } \DefFm \!\in\! \Subfmlas{\DefBseqVar}}
    \!\in\!
    \TotalSubMVSet{\DefBMatrix}$.
\end{enumerate}
\end{replemma}
\begin{proof}
The strategy to prove each item is the same: by contraposition, use the data from the assumptions
to compose an instance of a rule schema of the corresponding group of rule schemas. 
%
We detail below the proof for the third item.
Suppose that there is a connective
$\DefSymbol \in \SigAritySet{\Sigma}{k}$,
a formula
    $\DefFm \SymbDef \DefSymbol(\DefFm_1,\ldots,\DefFm_k) \in \Subfmlas{\DefBseqVar}$, a sequence $(\DefValueVar_1,\ldots,\DefValueVar_k)$ of truth-values
    with
    $\DiscriminatorVal{\DefCogVar}{\DefCogVar}{\DefValueVar_i}{\DefBMatrix}(\DefFm_i) \subseteq \CutPropSet_\DeffCogVar$
    and $\DiscriminatorVal{\DefCogVar}{\InvCog{\DefCogVar}}{\DefValueVar_i}{\DefBMatrix}(\DefFm_i) \subseteq \SetCompl{\CutPropSet_\DeffCogVar}$ for
    each $1 \leq i \leq k$ and $\Tuple{\DefCogVar,\DeffCogVar} \in \Set{\Tuple{\Acc,\TheV}, \Tuple{\Rej,\TheNV}}$, and some $\DeffValueVar \not\in \BMatrixInterp{\DefSymbol}{\DefBMatrix}(\DefValueVar_1,\ldots,\DefValueVar_k)$
    such that
    $\DiscriminatorVal{\DefCogVar}{\DefCogVar}{\DeffValueVar}{\DefBMatrix}(\DefFm) \subseteq \CutPropSet_\DeffCogVar$
    and $\DiscriminatorVal{\DefCogVar}{\InvCog{\DefCogVar}}{\DeffValueVar}{\DefBMatrix}(\DefFm) \subseteq \SetCompl{\CutPropSet_\DeffCogVar}$ for each $\Tuple{\DefCogVar,\DeffCogVar} \in \Set{\Tuple{\Acc,\TheV}, \Tuple{\Rej,\TheNV}}$.
    Then
    $
        \bigcup_{1 \leq i \leq k}
        \DiscriminatorVal{\DefCogVar}{\DefCogVar}{\DefValueVar_i}{\DefBMatrix}(\DefFm_i) 
        \cup \DiscriminatorVal{\DefCogVar}{\DefCogVar}{\DeffValueVar}{\DefBMatrix}(\DefFm)
        \subseteq \CutPropSet_\DeffCogVar \cap 
        \GenSubfmlas{}{\Discriminator{\DefBMatrix}}(\DefBseqVar)
        \text{ and }
        \bigcup_{1 \leq i \leq k}
        \DiscriminatorVal{\DefCogVar}{\InvCog{\DefCogVar}}{\DefValueVar_i}{\DefBMatrix}(\DefFm_i) 
        \cup \DiscriminatorVal{\DefCogVar}{\InvCog{\DefCogVar}}{\DeffValueVar}{\DefBMatrix}(\DefFm)
        \subseteq \SetCompl{\CutPropSet_\DeffCogVar} \cap 
        \GenSubfmlas{}{\Discriminator{\DefBMatrix}}(\DefBseqVar)
    $
    for each $\Tuple{\DefCogVar,\DeffCogVar} \in \Set{\Tuple{\Acc,\TheV}, \Tuple{\Rej,\TheNV}}$, and thus
    we have
    $\BAnaCalcCon{\IntermAccSet{\CutPropSet}}{\SetCompl{\IntermRejSet{\CutPropSet}}}{\IntermRejSet{\CutPropSet}}{\SetCompl{\IntermAccSet{\CutPropSet}}}{\BIndAxiomatRuleSigmaName}{\GenSubfmlas{}{\Discriminator{\DefBMatrix}}}$.
\end{proof}

\begin{reptheorem}{the:completeness}
If $\Discriminator{\DefBMatrix}$ is a discriminator for
a \upshape{$\theBPN[\Sigma]$}--matrix
$\DefBMatrix$, then
the calculus
$\BIndAxiomatName{\DefBMatrix}{\Discriminator{\DefBMatrix}}{}$
is complete with respect to
$\DefBMatrix$. 
Furthermore, this calculus is $\SeparatorsSet{\DefBMatrix}$--analytic.
\end{reptheorem}
\begin{proof}
Let $\DefBseqVar \SymbDef \BStat{\DefAccSet}{\DefNRejSet}{\DefRejSet}{\DefNAccSet}$ 
be a \theB-statement
and
suppose that
(a) $\nBAnaCalcCon{\DefAccSet}{\DefNRejSet}{\DefRejSet}{\DefNAccSet}{\BIndAxiomatName{\DefBMatrix}{\Discriminator{\DefBMatrix}}{}}{\GenSubfmlas{}{\Discriminator{\DefBMatrix}}}$. Our goal is to build an
$\DefBMatrix$-valuation witnessing
$\nBEnt{\DefAccSet}{\DefNRejSet}{\DefRejSet}{\DefNAccSet}{\DefBMatrix}$.
From (a), by $\ref{prop:BConC}$,
we have that (b) there are
$\DefAccSet \subseteq \IntermAccSet{\CutPropSet} \subseteq \SetCompl{\DefNAccSet}$
and
$\DefRejSet \subseteq \IntermRejSet{\CutPropSet} \subseteq \SetCompl{\DefNRejSet}$ such that
$\nBAnaCalcCon{\IntermAccSet{\CutPropSet}}{\SetCompl{\IntermRejSet{\CutPropSet}}}{\IntermRejSet{\CutPropSet}}{\SetCompl{\IntermAccSet{\CutPropSet}}}{\BIndAxiomatName{\DefBMatrix}{\Discriminator{\DefBMatrix}}{}}{\GenSubfmlas{}{\Discriminator{\DefBMatrix}}}$. Consider then a mapping
$f : \Subfmlas{\DefBseqVar} \to \DefVSet$
with (c)
$f(\DefFm) \in \BMatrixDistSet{\DefCogVar}{\DefBMatrix}$
    if{f} $\DefFm \in \CutPropSet_\DeffCogVar$, for
    $\Tuple{\DefCogVar,\DeffCogVar} \in \Set{\Tuple{\Acc,\TheV}, \Tuple{\Rej,\TheNV}}$,
    whose existence is guaranteed by
    items (1) and (2) of
    Lemma \ref{lem:auxcompleteness}.
Notice that items (3) and (4)
of this same lemma imply, respectively, that $f(\DefSymbol(\DefFm_1, \ldots, \DefFm_k)) \in \BMatrixInterp{\DefSymbol}{\DefBMatrix}(
f(\DefFm_1), \ldots, f(\DefFm_k))$
for every $\DefSymbol(\DefFm_1, \ldots, \DefFm_k) \in \GenSubfmlas{}{\Discriminator{\DefBMatrix}}(\DefBseqVar)$,
and $f(\Subfmlas{\DefBseqVar}) \in \TotalSubMVSet{\DefBMatrix}$.
Hence, $f$ may be extended to 
an $\DefBMatrix$-valuation $\DefVal$
and, from (b) and (c), we have
$\DefVal(\CogSet{\DefFmlaSet}{\DefCogVar})
        \subseteq \BMatrixDistSet{\DefCogVar}{\DefBMatrix}$
        for each $\DefCogVar \in \CogsSet$,
so $\nBEnt{\DefAccSet}{\DefNRejSet}{\DefRejSet}{\DefNAccSet}{\DefBMatrix}$.
\end{proof}

The calculi presented so far (Examples \ref{ex:axiomatinfo}
and \ref{ex:otheraxiomats})
were produced by means of the axiomatization procedure just described, followed by some simplifications consisting of
removing instances of conditions 
\ref{prop:BConO} and
\ref{prop:BConD}, and using condition \ref{prop:BConC} 
on pairs of schemas having the forms $\TwoDRule{\DefAccSet,\DefFm}{\DefNRejSet}{\DefRejSet}{\DefNAccSet}$ and $\TwoDRule{\DefAccSet}{\DefNRejSet}{\DefRejSet}{\DefNAccSet,\DefFm}$, or the forms $\TwoDRule{\DefAccSet}{\DefNRejSet,\DefFm}{\DefRejSet}{\DefNAccSet}$ and $\TwoDRule{\DefAccSet}{\DefNRejSet}{\DefRejSet,\DefFm}{\DefNAccSet}$, yielding in either case 
the schema \mbox{$\TwoDRule{\DefAccSet}{\DefNRejSet}{\DefRejSet}{\DefNAccSet}$}. By
Theorem \ref{the:completeness} and the fact that these simplifications preserve analyticity, it follows
that such calculi are analytic.
It is also worth mentioning
that this same procedure may be applied
to the matrix $\MCIFiveVMatrix$
in view of its monadicity (see a discriminator for it in Example \ref{ex:disc}), which means
that we also obtain a \emph{finite} Hilbert-style
symmetrical 
axiomatization for $\textbf{mCi}$.

%% file: tex/proofsearch.tex
\section{Proof search in Two Dimensions}
\label{sec:proofsearch}

Throughout this section, 
let $\DefBseqVar \SymbDef \BStat{\DefAccSet}{\DefNRejSet}{\DefRejSet}{\DefNAccSet}$
be an arbitrary \theB-sequent,
$\DefCalculusName$ be a finite and finitary calculus,
and $\DeffFmlaSet$ be a finite set of formulas.
Notice that, whenever $\DefCalculusName$ is $\DeffFmlaSet$-analytic, it is enough to consider the rule instances in
 $\RuleInstSetCalcSeq{\DefCalculusName}{\DefBseqVar}$
in order to provide a proof of~$\DefBseqVar$ in $\DefCalculusName$.
Searching for such a proof is clearly
a particular case of finding a proof of~$\DefBseqVar$ using only candidates in a finite set $R$ of finitary rule instances. 
A~proof-search algorithm for this more general
setting 
is presented in
Algorithm~\ref{algo:proofsearch} by means of a function called $\ProofSAlgoName$. 
The algorithm searches for a proof by expanding nodes that
are not closed or discontinued using only instances in~$R$
that were not used yet in the branch of the node under expansion.
As we shall see in the sequel, the order in which
applicable instances are selected does not affect
the result, although for sure smarter choice heuristics may well improve the performance of the algorithm in particular cases.%

\begin{algorithm}
\small
\caption{Proof search over a finite set of finitary rule instances}
\label{algo:proofsearch}
\DontPrintSemicolon
\SetKwFunction{FExpand}{\textrm\ProofSAlgoName}
\SetKwFunction{FNode}{Node}
\SetKwFunction{FClosedNode}{OneNodeTree}
\SetKwFunction{FNode}{OneNodeTree}
\SetKwProg{Fn}{function}{:}{}
\Fn{\FExpand{$F\SymbDef\Tuple{\DeffFmlaSet_\Acc,\DeffFmlaSet_\Rej}$, $C\SymbDef\Tuple{\DefNRejSet, \DefNAccSet}$, $R$}}{
    \KwIn{antecedents in $F$, succedents in $C$ and a finite set $R$ of finitary rule instances}
        $\DefTreeVar \gets \SingleNodeTree(F)$ \label{algoline:maket}\\
        \llIf{$\DeffFmlaSet_\DefCogVar \cap \DefFmlaSet_{\InvCog\DefCogVar} \neq \EmptySet$ for some $\DefCogVar \in \Set{\Acc,\Rej}$}{\Return{$\DefTreeVar$}}\label{algoline:base}\\
        \ForEach{rule instance $\DefRuleInstVar \SymbDef \TwoDRule{\DeffffAccSet}{\DeffffNRejSet}{\DeffffRejSet}{\DeffffNAccSet} \in R$ 
        \label{algoline:looprules}}{
          \If{$\DeffffFmlaSet_{\InvCog\DefCogVar} \cap \DeffFmlaSet_{\DefCogVar} = \EmptySet$
            and $\DeffffFmlaSet_{\DefCogVar} \subseteq \DeffFmlaSet_{\DefCogVar}$
            for each $\DefCogVar \in \Set{\Acc,\Rej}$\label{algoline:condpremiss}}{
                \lIf{$\DeffffFmlaSet_\NAcc \cup \DeffffFmlaSet_\NRej = \EmptySet$\label{algoline:starcond}}{\Return{$\DefTreeVar$ with a single child $\SingleNodeTree(\StarLabel)$}}
                \ForEach{$\DefCogVar \in \Set{\Acc,\Rej}$ and $\DefFm \in \DeffffFmlaSet_{\InvCog{\DefCogVar}}$\label{algoline:loopfmlas}}{
                    $\DefTreeVar^\prime \gets\, $\FExpand{$\Tuple{\DeffFmlaSet_\Acc \cup P_\Acc(\DefFm),\DeffFmlaSet_\Rej \cup P_\Rej(\DefFm)}$, $C$, $\SetDiff{R}{\Set{\DefRuleInstVar}}$},
                    where
                    $P_\DefCogVar(\DefFm)$ is $\EmptySet$ if $\DefFm \not\in \DeffffFmlaSet_\DefCogVar$ and $\Set{\DefFm}$ otherwise \label{algoline:reccall}\\
                    \textbf{add} $\TreeRoot{\DefTreeVar^\prime}$ as a child of $\TreeRoot{\DefTreeVar}$ in $\DefTreeVar$\label{algoline:addchild}\\
                    \lIf{$\DefTreeVar^\prime$ is not $C$-closed}{\Return{$\DefTreeVar$}}\label{algoline:notcclosed}
                }
                \lIf{$\DefTreeVar$ is $C$-closed
                \label{algoline:closedcond}}{\Return{$\DefTreeVar$}}
            }
        }
        \Return{$\DefTreeVar$\label{algoline:opentree}}
    }
\end{algorithm}

The following lemma 
(verifiable by induction on the size of $R$) 
proves the termination of
$\ProofSAlgoName$ and its correctness.
The subsequent result establishes the
applicability of this algorithm for proof search
over $\DeffFmlaSet$-analytic
calculi.

\begin{replemma}{lem:correctness}
Let
$R$ be a finite set of finitary rule instances. 
Then
the procedure 
$\ProofSAlgoName(
\Tuple{\DefAccSet, \DefRejSet},
\Tuple{\DefNAccSet, \DefNRejSet}, 
R)$
always terminates, returning~a tree 
that is $\Tuple{\DefNAccSet, \DefNRejSet}$-closed
iff
$\BRuleInstCon{\DefAccSet}{\DefNRejSet}{\DefRejSet}{\DefNAccSet}{R}$.
\end{replemma}

\begin{replemma}{lem:anacalcproofs}
If $\DefCalculusName$ is $\DeffFmlaSet$-analytic, then
$\ProofSAlgoName$ is a proof-search algorithm for $\DefCalculusName$ and a decision procedure for
$\BCalcConName{\DefCalculusName}$.
\end{replemma}
\begin{proof}
We know that $\RuleInstSetCalcSeq{\DefCalculusName}{\DefBseqVar}$
provides enough material for a derivation
of~$\DefBseqVar$ to be produced, since $\DefCalculusName$ is $\DeffFmlaSet$-analytic.
Clearly, such set is finite and contains only
finitary rule instances, hence
the present result is a direct consequence of
Lemma~\ref{lem:correctness}. 
\end{proof}

The next results concern the complexity
 of Algorithm \ref{algo:proofsearch}. In what follows, let $R$ be a finite set of finitary rule instances, $b \SymbDef \max_{\DefRuleInstVar \in R} \BranchSymb(\DefRuleInstVar)$, 
$s \SymbDef \SeqSizeSymb(\Set{\DefBseqVar}\cup R)$
and $n \SymbDef \SetSize{R}$.
We shall use 
$\Polynomial{}(m)$ to refer to
``a polynomial in $m$''.

\begin{replemma}{lem:complexity}
The worst-case running time of $\ProofSAlgoName(
\Tuple{\DefAccSet, \DefRejSet},
\Tuple{\DefNAccSet, \DefNRejSet}, 
R)$ is $O(b^{n} + n \cdot \Polynomial{}(s))$.
\end{replemma}
\begin{proof}
Let $T(n, s)$ be the worst-case running-time of $\ProofSAlgoName$. Note that it occurs under three conditions:
 first, $\BRuleInstCon{\DefAccSet}{\DefNRejSet}{\DefRejSet}{\DefNAccSet}{R}$; second,
 the set $R$ needs to be entirely inspected until an applicable rule instance is found; and third, such an instance does not have an empty set of succedents.
Notice that $T(0,s) = c_1 + \Polynomial{}(s)$ and, based on the
assignments above and after some algebraic manipulations, we have, for $n \geq 1$,
    $T(n, s) \leq b \cdot T(n-1, s + \Polynomial{}(s)) + 2n \cdot \Polynomial{}(s)$.
It is then straightforward to check by induction on $n$
that $T(n,s) \in O(b^n + n \cdot \Polynomial{}(s))$.
\end{proof}

\begin{reptheorem}{the:analyticexp}
If $\DefCalculusName$ is $\DeffFmlaSet$-analytic, $\ProofSAlgoName$
is a proof-search algorithm
for $\DefCalculusName$
that runs in exponential time in general, and in 
polynomial time
if $\DefCalculusName$ contains only rules with at most one formula in the succedent.
\end{reptheorem}
\begin{proof}
Clearly, the set of all instances
of rules of $\DefCalculusName$
using only 
formulas in 
$\GenSubfmlas{}{\DeffFmlaSet}(\DefBseqVar)$ is finite and contains only
finitary rule instances, and its size is polynomial in $\SeqSizeSymb(\DefBseqVar)$. 
The announced result then follows
directly from Lemma \ref{lem:complexity}.
\end{proof}

The previous result makes the axiomatization procedure presented in
Section~\ref{sec:axiomatization}
even more attractive, since it
delivers a 
$\SeparatorsSet{\DefBMatrix}$--analytic calculus for $\DefBMatrix$,
where $\SeparatorsSet{\DefBMatrix}$ is
a finite set of formulas acting as separators. It follows then that
$\ProofSAlgoName$
is a proof-search 
algorithm
for such axiomatization running in at most
exponential time.
More than that,
$\ProofSAlgoName$ outputs a
tree with at least one open branch
when the \theB-sequent $\DefBseqVar$ of interest
is not provable. From such branch, one may obtain a partition
of $\GenSubfmlas{}{\Discriminator{\DefBMatrix}}(\DefBseqVar)$ and, by
Proposition \ref{lem:auxcompleteness},
define a mapping on $\Subfmlas{\DefBseqVar}$
that extends to an $\DefBMatrix$-valuation.
It follows that the discussed algorithm
may easily be adapted so as to deliver a countermodel when $\DefBseqVar$ is unprovable. For experimenting with the
axiomatization procedure and searching for proofs over the generated
calculus, one can 
make use of the implementation that may be found at
\url{https://github.com/greati/logicantsy}.
We should also emphasize that, by Theorem~\ref{the:analyticexp} and the axiomatization procedure given in Section~\ref{sec:axiomatization}, we have:
\begin{corollary}
Any finite monadic
$\theBPN[\Sigma]$--matrix $\DefBMatrix$
whose induced axiomatization
contains only rules with at most
one succedent
is decidable in polynomial time.
\end{corollary}

\noindent By the above result, then, the \theB-entailment relation $\BEnt{\PlaceholderArg}{\PlaceholderArg}{\PlaceholderArg}{\PlaceholderArg}{\InfoFourMatrix}$ (from Example~{\ref{ex:infomatrix}}) is decidable in polynomial time.
Consequently, the same also holds for its $\tAsp$-aspect, which is inhabited by the 4-valued logic introduced in {\upshape\cite{avron2007}}.

In addition, it is worth stressing that, although no better in the limiting cases, the axiomatization provided in Section~\ref{sec:axiomatization} together with the algorithm presented in this section translate the problem of deciding a \theB-entailment relation into a purely symbolic procedure that may perform better than searching for $\DefBMatrix$-valuations in some cases.

We close with another complexity result concerning the decidability
of $\BCalcConName{\DefCalculusName}$, 
complementing the one 
given by the discussed algorithm; it follows
by an argument similar to the one
presented for the one-dimensional case
in \cite{marcelinocaleiro2016}.

\begin{reptheorem}{the:coNPdec}
If $\DefCalculusName$ is $\DeffFmlaSet$-analytic, then the problem of deciding $\BCalcConName{\DefCalculusName}$ is in $\coNP$.
\end{reptheorem}

%% file: tex/conclusion.tex
\section{Conclusion}
\label{sec:conclusion}

In this paper, we approached bilateralism by exploring a two-dimensional notion of consequence, considering the cognitive attitudes of acceptance and rejection instead of the conventional speech acts of assertion and of denial.
Our intervention has been two-fold: 
on the semantical front we have employed two-dimensional (partial) non-deterministic logical matrices, 
and on proof-theoretical grounds we have employed two-dimensional symmetrical proof formalisms which generalize traditional Hilbert-style calculi and their associated unilinear notion of derivation.
As a result, and generalizing~\cite{marcelinowollic}, we have provided
an axiomatization procedure that delivers analytic calculi for a very expressive class of finite monadic matrices. 
On what concerns proof development, in spite of well-known evidence about the p-equivalence between Hilbert-style calculi and Gentzen-style calculi (\cite{CookReckhow}), die-hard popular belief concerning their `deep inequivalence' seems hard to wash away.  To counter that belief with facts, we developed for our calculi a general proof-search algorithm that was secured to run in exponential time.%

We highlight that our two-dimensional proof-formalism
differs in important respects from the many-placed sequent calculi
used in~\cite{avron2005igpl} to axiomatize (one-dimensional total) non-deterministic matrices (requiring no sufficient expressiveness) and in~\cite{ole2014} for
approaching multilateralism.
First, a many-placed sequent calculus is not Hilbert-style: 
rules manipulate complex objects
whose structure involve contexts and considerably deviate from the shape
of the consequence relation being captured; our calculi, on the other hand, 
are contained in
their corresponding \theB-consequences.
Second, when axiomatizing a matrix, 
the structure of many-placed sequents grows according to
the number of values ($n$ places for $n$ truth-values); our rule schemas, in turn, remain with four places, 
and reflect
the complexity of the underlying semantics in the complexity of the formulas being manipulated. 
Moreover, the study
of many-placed sequents currently contemplates only one-dimensional
consequence relations; extending them to the two-dimensional case is a line of research worth exploring.%

As further future work, we envisage generalizing the two-dimensional notion of consequence relation by allowing logics over different languages (\cite{Humberstone-HL}) 
---for instance, conflating different logics or different fragments of some given logic of interest--- 
to coinhabit the same logical structure, each one along its own dimension, while controlling their interaction at the object-language level, taking advantage of the framework and the results in~\cite{marcelinocaleiro2016}. 
This opens the doors for a line of investigation on whether or to what extent the individual characteristics of these ingredient logics, such as their decidability status, may be preserved. 
With respect to our proof-search algorithm, an important research path to be explored would involve the design of 
heuristics for smarter choices of rule instances used to expand nodes during the search, as this
may improve the performance of the algorithm on certain classes of logics. 
At last, we also expect to extend 
the present research so as 
to cover multidimensional notions of consequence, in order to provide increasingly general technical and philosophical grounds for the study of logical pluralism.


%% file: tex/appendix.tex
\section{Proofs of theoretical results}
\label{sec:appendix}


\begin{lemma}
\label{lem:caracvalues}
If $\DefBMatrix$ is a
monadic {\upshape$\theBPN[\Sigma]$}--matrix
and $\Discriminator{\DefBMatrix}$
is a discriminator for $\DefBMatrix$,
then, for all $\DefFm \in \DefLangSet$, 
$\DefValueVar \in \DefVSet$
and $\DefBMatrix$-valuation
$\DefVal$,
\begin{equation*}
    \DefVal(\DefFm) = \DefValueVar
    \text{\ \ if{f}\ \ }
    \DefVal(\DiscriminatorVal{\DefCogVar}{\DefCogVar}{\DefValueVar}{\DefBMatrix}(\DefFm))
    \subseteq \BMatrixDistSet{\DefCogVar}{\DefBMatrix}
    \text{ and }
    \DefVal(\DiscriminatorVal{\DefCogVar}{\InvCog{\DefCogVar}}{\DefValueVar}{\DefBMatrix}(\DefFm))
    \subseteq \BMatrixDistSet{{\InvCog{\DefCogVar}}}{\DefBMatrix}
    \text{ for every }
    \DefCogVar \in \Set{\Acc,\Rej}.
\end{equation*}
\end{lemma}
\begin{proof}
From the left to the right,
assume that $\DefVal(\DefFm) = \DefValueVar$
and let $\DefCogVar \in \Set{\Acc,\Rej}$.
If $\DefSeparator \in \DiscriminatorVal{\DefCogVar}{\DefCogVar}{\DefValueVar}{\DefBMatrix}$,
then 
$
\DefVal(\DefSeparator(\DefFm))
\in \BMatrixInterpFm{\DefSeparator}{\DefBMatrix}(\DefVal(\DefFm)) 
= \BMatrixInterpFm{\DefSeparator}{\DefBMatrix}(\DefValueVar)$,
and we know that $\BMatrixInterpFm{\DefSeparator}{\DefBMatrix}(\DefValueVar) \subseteq \BMatrixDistSet{\DefCogVar}{\DefBMatrix}$ if
$\DefSeparator \in \DiscriminatorVal{\DefCogVar}{\DefCogVar}{\DefValueVar}{\DefBMatrix}$.
The same reasoning applies
for $\DefSeparator \in \DiscriminatorVal{\InvCog\DefCogVar}{\InvCog\DefCogVar}{\DefValueVar}{\DefBMatrix}$.
Conversely, we may argue contrapositively:
suppose that
$\DefVal(\DefFm) = \DeffValueVar \neq \DefValueVar$
and consider the $\DefCogVar$-separator
$\DefSeparator$
for $\DefValueVar$
and $\DeffValueVar$, for some
$\DefCogVar \in \Set{\Acc,\Rej}$.
By cases,
if $\DefSeparator \in \DiscriminatorVal{\DefCogVar}{\DefCogVar}{\DefValueVar}{\DefBMatrix}$,
then 
$\DefVal(\DefSeparator(\DefFm)) \in \BMatrixInterpFm{\DefSeparator}{\DefBMatrix}(\DeffValueVar) \subseteq \BMatrixDistSet{\InvCog{\DefCogVar}}{\DefBMatrix}$
and so 
$\DefVal(\DiscriminatorVal{\DefCogVar}{\DefCogVar}{\DefValueVar}{\DefBMatrix}(\DefFm))
    \not\subseteq \BMatrixDistSet{\DefCogVar}{\DefBMatrix}$; analogously, we have
 $\DefVal(\DiscriminatorVal{\DefCogVar}{\InvCog{\DefCogVar}}{\DefValueVar}{\DefBMatrix}(\DefFm))
    \not\subseteq \BMatrixDistSet{{\InvCog{\DefCogVar}}}{\DefBMatrix}$ 
    if
$\DefSeparator \in \DiscriminatorVal{\DefCogVar}{\InvCog{\DefCogVar}}{\DefValueVar}{\DefBMatrix}$.
\end{proof}
\begin{theorem}
\label{the:soundness}
If $\Discriminator{\DefBMatrix}$ is a discriminator for
a $\theBPN[\Sigma]$--matrix
$\DefBMatrix$, then
the calculus
$\BIndAxiomatName{\DefBMatrix}{\Discriminator{\DefBMatrix}}{}$
is sound with respect to
$\DefBMatrix$.
\end{theorem}
\begin{proof}
We will show that any
$\theBPN[\Sigma]$--valuation
that constituted a countermodel for
a schema of $\BIndAxiomatName{\DefBMatrix}{\Discriminator{\DefBMatrix}}{}$
would lead to a contradiction.
The argument will cover each of the groups of schemas of the concerned calculus.

\begin{namedproperties}
    
    \item[\ref{calc:exists}] Consider a schema
    $\DefBseqVar \SymbDef \TwoDRule
        {\DiscriminatorValCh{\Acc}{\NAcc}{\DefSetValueVarProp}{\DefBMatrix}}
        {\DiscriminatorValCh{\Rej}{\Rej}{\VSetComp{\DefSetValueVar}}{\DefBMatrix}}
        {\DiscriminatorValCh{\Rej}{\NRej}{\DefSetValueVarProp}{\DefBMatrix}}
        {\DiscriminatorValCh{\Acc}{\Acc}{\VSetComp{\DefSetValueVar}}{\DefBMatrix}}$,
        for some $\DefSetValueVarProp \subseteq \DefVSet$
        and some choice of
         $\Tuple{\DiscriminatorValCh{\Acc}{\Acc}{\VSetComp\DefSetValueVar}{\DefBMatrix},\DiscriminatorValCh{\Rej}{\Rej}{\VSetComp\DefSetValueVar}{\DefBMatrix}}$ and $\Tuple{\DiscriminatorValCh{\Acc}{\NAcc}{\DefSetValueVarProp}{\DefBMatrix},\DiscriminatorValCh{\Rej}{\NRej}{\DefSetValueVarProp}{\DefBMatrix}}$.
         Suppose that 
         $\DefBseqVar$
         does not hold in $\DefBMatrix$, 
         with the valuation $\DefVal$ witnessing this fact. We will prove that,
         given a propositional variable
         $\DefProp$, $\DefVal(\DefProp) \neq \DefValueVar$, for all $\DefValueVar \in \DefVSet$, an absurd.
         For that purpose, let $\DefValueVar \in \DefVSet$.
         In case $\DefValueVar \in \DefSetValueVarProp$,
         there must be a separator
         $\DefSeparator$ in
         $\DiscriminatorVal{\DefCogVar}{\InvCog{\DefCogVar}}{\DefValueVar}{\DefBMatrix}$, for some $\DefCogVar \in \Set{\Acc,\Rej}$, such that
         $\DefVal(\DefSeparator(\DefProp)) \in \BMatrixDistSet{\DefCogVar}{\DefBMatrix}$.
         By Lemma \ref{lem:caracvalues}, this implies 
         that $\DefVal(\DefProp) \neq \DefValueVar$.
         The reasoning is similar in case $\DefValueVar \in \VSetComp{\DefSetValueVar}$.
    \item[\ref{calc:d}]
    Let $\DefProp \in \DefProps$
    and $\DefValueVar \in \DefVSet$
    be such that $\DefValueVar \in \BMatrixDistSet{\DefCogVar}{\DefBMatrix}$, 
    with $\DefCogVar \in \Set{\Acc,\NAcc}$.
    Suppose that there is an
         $\DefBMatrix$-valuation
         $\DefVal$ under which the
         schema
        $\TwoDRule
        {\DiscriminatorVal{\Acc}{\Acc}{\DefValueVar}{\DefBMatrix}(\DefProp), \FmlaPlacerByValue{\NAcc}(\DefValueVar)}
        {\DiscriminatorVal{\Rej}{\NRej}{\DefValueVar}{\DefBMatrix}(\DefProp)}
        {\DiscriminatorVal{\Rej}{\Rej}{\DefValueVar}{\DefBMatrix}(\DefProp)}
        {\DiscriminatorVal{\Acc}{\NAcc}{\DefValueVar}{\DefBMatrix}(\DefProp), \FmlaPlacerByValue{\Acc}(\DefValueVar)}$
        does not hold. 
        Then, $\DefVal(\FmlaPlacerByValue{\DefCogVar}(\DefValueVar)) \subseteq \BMatrixDistSet{\InvCog{\DefCogVar}}{\DefBMatrix}$ and thus, since
        $\DefValueVar \in \BMatrixDistSet{\DefCogVar}{\DefBMatrix}$, it follows that $\DefVal(\DefProp) \in \BMatrixDistSet{\InvCog{\DefCogVar}}{\DefBMatrix}$. 
        On the other hand,
        since
        $\DefVal(\DiscriminatorVal{\DefCogVar}{\DefCogVar}{\DefValueVar}{\DefBMatrix}(\DefProp)) \subseteq \BMatrixDistSet{\DefCogVar}{\DefBMatrix}$ and
        $\DefVal(\DiscriminatorVal{\DefCogVar}{\InvCog{\DefCogVar}}{\DefValueVar}{\DefBMatrix}(\DefProp)) \subseteq \BMatrixDistSet{\InvCog{\DefCogVar}}{\DefBMatrix}$
        for each $\DefCogVar \in \Set{\Acc,\Rej}$,
        by Lemma \ref{lem:caracvalues}
        we have that $\DefVal(\DefProp) = \DefValueVar$, a contradiction.
        The proof for the other schema
        is analogous.
    \item[\ref{calc:sigma}]\sloppy 
    Let $\DefSymbol \in \SigAritySet{\Sigma}{k}$,
    $X = \Family{\DefValueVar_i}{i=1}^k$ be a family of
    truth-values of $\DefBMatrix$, $\DeffValueVar \not\in \BMatrixInterp{\DefSymbol}{\DefBMatrix}(\DefValueVar_1,\ldots,\DefValueVar_k)$
    and
    let $\Family{\DefProp_i}{i=1}^k$
    be a family of distinct propositional variables.
    Suppose the schema
    $  \TwoDRule
        {\SigmaRulesSet{\Acc}{\DefSymbol, X,\DeffValueVar}}
        {\SigmaRulesSet{\NRej}{\DefSymbol, X,\DeffValueVar}}
        {\SigmaRulesSet{\Rej}{\DefSymbol, X,\DeffValueVar}}
        {\SigmaRulesSet{\NAcc}{\DefSymbol, X,\DeffValueVar}}$ does not hold under
    $\DefVal$. By Lemma \ref{lem:caracvalues}, we have that
    $\DefVal(\DefProp_i)=\DefValueVar_i$,
    for all $1 \leq i \leq k$,
    and
    $\DefVal(\DefSymbol(\DefProp_1,\ldots,\DefProp_k)) = \DeffValueVar$.
    It follows that $\DeffValueVar = \DefVal(\DefSymbol(\DefProp_1,\ldots,\DefProp_k)) \in \BMatrixInterp{\DefSymbol}{\DefBMatrix}(\DefVal(\DefProp_1),\ldots,\DefVal(\DefProp_k)) = \BMatrixInterp{\DefSymbol}{\DefBMatrix}(\DefValueVar_1,\ldots,\DefValueVar_k)$, contradicting one of the 
    assumptions.
    \item[\ref{calc:nontotal}]
    Assume that $\DefSetValueVar \not\in \TotalSubMVSet{\DefBMatrix}$
    and let
    $\Family{\DefProp_\DefValueVar}{\DefValueVar \in \DefSetValueVar}$
    be a family of distinct propositional variables. If 
    the schema $        \TwoDRule
        {\bigcup_{\DefValueVar \in \DefSetValueVar}\DiscriminatorVal{\Acc}{\Acc}{\DefValueVar}{\DefBMatrix}(\DefProp_\DefValueVar)}
        {\bigcup_{\DefValueVar \in \DefSetValueVar}\DiscriminatorVal{\Rej}{\NRej}{\DefValueVar}{\DefBMatrix}(\DefProp_\DefValueVar)}
        {\bigcup_{\DefValueVar \in \DefSetValueVar}\DiscriminatorVal{\Rej}{\Rej}{\DefValueVar}{\DefBMatrix}(\DefProp_\DefValueVar)}
        {\bigcup_{\DefValueVar \in \DefSetValueVar}\DiscriminatorVal{\Acc}{\NAcc}{\DefValueVar}{\DefBMatrix}(\DefProp_\DefValueVar)}$
        does not hold under
        an $\DefBMatrix$-valuation
        $\DefVal$,
        then Lemma \ref{lem:caracvalues}
        guarantees that
        $\DefVal(\DefProp_\DefValueVar) = \DefValueVar$
        for each $\DefValueVar \in \DefSetValueVar$.
        Hence, since $\DefProp_\DefValueVar \in \DefProps$ for each $\DefValueVar \in \DefSetValueVar$, we have $\DefSetValueVar \subseteq \DefVal(\DefLangSet) \in \TotalSubMVSet{\DefBMatrix}$,
        contradicting the assumption.
\end{namedproperties}
\end{proof}

\begin{lemma}
\label{lem:auxcompleteness}
Let $\Discriminator{\DefBMatrix}$ be a discriminator for
a $\theBPN[\Sigma]$--matrix
$\DefBMatrix$
and denote by
$\GenSubfmlas{}{\Discriminator{\DefBMatrix}}$
the mapping $\GenSubfmlas{}{\SeparatorsSet{\DefBMatrix}}$. For all \theB-statements $\DefBseqVar$ of the form $ \BStat{\IntermAccSet{\CutPropSet}}{\SetCompl{\IntermRejSet{\CutPropSet}}}{\IntermRejSet{\CutPropSet}}{\SetCompl{\IntermAccSet{\CutPropSet}}}${\upshape :}
\vspace{-2mm}
\begin{enumerate}
    \item if $\nBAnaCalcCon{\IntermAccSet{\CutPropSet}}{\SetCompl{\IntermRejSet{\CutPropSet}}}{\IntermRejSet{\CutPropSet}}{\SetCompl{\IntermAccSet{\CutPropSet}}}{\BIndAxiomatRuleExistsName}{\GenSubfmlas{}{\Discriminator{\DefBMatrix}}}$,
    then for all $\DefFm \in \Subfmlas{\DefBseqVar}$
    there is an $\DefValueVar \in \DefVSet$ such that
    $\DiscriminatorVal{\DefCogVar}{\DefCogVar}{\DefValueVar}{\DefBMatrix}(\DefFm) \subseteq \CutPropSet_\DeffCogVar$
    and $\DiscriminatorVal{\DefCogVar}{\InvCog{\DefCogVar}}{\DefValueVar}{\DefBMatrix}(\DefFm) \subseteq \SetCompl{\CutPropSet_\DeffCogVar}$, for
    $\Tuple{\DefCogVar,\DeffCogVar} \in \Set{\Tuple{\Acc,\TheV}, \Tuple{\Rej,\TheNV}}$;
    \item if $\nBAnaCalcCon{\IntermAccSet{\CutPropSet}}{\SetCompl{\IntermRejSet{\CutPropSet}}}{\IntermRejSet{\CutPropSet}}{\SetCompl{\IntermAccSet{\CutPropSet}}}{\BIndAxiomatRuleDName}{\GenSubfmlas{}{\Discriminator{\DefBMatrix}}}$,
    then for every $\DefFm \in \Subfmlas{\DefBseqVar}$
    and $\DefValueVar \in \DefVSet$ such that
    $\DiscriminatorVal{\DefCogVar}{\DefCogVar}{\DefValueVar}{\DefBMatrix}(\DefFm) \subseteq \CutPropSet_\DeffCogVar$
    and $\DiscriminatorVal{\DefCogVar}{\InvCog{\DefCogVar}}{\DefValueVar}{\DefBMatrix}(\DefFm) \subseteq \SetCompl{\CutPropSet_\DeffCogVar}$, we have
    $\DefValueVar \in \BMatrixDistSet{\DefCogVar}{\DefBMatrix}$
    if{f} $\DefFm \in \CutPropSet_\DeffCogVar$, for
    $\Tuple{\DefCogVar,\DeffCogVar} \in \Set{\Tuple{\Acc,\TheV}, \Tuple{\Rej,\TheNV}}$;
    \item if $\nBAnaCalcCon{\IntermAccSet{\CutPropSet}}{\SetCompl{\IntermRejSet{\CutPropSet}}}{\IntermRejSet{\CutPropSet}}{\SetCompl{\IntermAccSet{\CutPropSet}}}{\BIndAxiomatRuleSigmaName}{\GenSubfmlas{}{\Discriminator{\DefBMatrix}}}$,
    then for every $\DefSymbol \in \SigAritySet{\Sigma}{k}$,
    $\DefFm \SymbDef \DefSymbol(\DefFm_1,\ldots,\DefFm_k) \in \Subfmlas{\DefBseqVar}$
    and $\DefValueVar_1,\ldots,\DefValueVar_k \in \DefVSet$
    with
    $\DiscriminatorVal{\DefCogVar}{\DefCogVar}{\DefValueVar_i}{\DefBMatrix}(\DefFm_i) \subseteq \CutPropSet_\DeffCogVar$
    and $\DiscriminatorVal{\DefCogVar}{\InvCog{\DefCogVar}}{\DefValueVar_i}{\DefBMatrix}(\DefFm_i) \subseteq \SetCompl{\CutPropSet_\DeffCogVar}$, for
    each $1 \leq i \leq k$ and $\Tuple{\DefCogVar,\DeffCogVar} \in \Set{\Tuple{\Acc,\TheV}, \Tuple{\Rej,\TheNV}}$, we have that
    $\DiscriminatorVal{\DefCogVar}{\DefCogVar}{\DeffValueVar}{\DefBMatrix}(\DefFm) \subseteq \CutPropSet_\DeffCogVar$
    and $\DiscriminatorVal{\DefCogVar}{\InvCog{\DefCogVar}}{\DeffValueVar}{\DefBMatrix}(\DefFm) \subseteq \SetCompl{\CutPropSet_\DeffCogVar}$ for each $\Tuple{\DefCogVar,\DeffCogVar} \in \Set{\Tuple{\Acc,\TheV}, \Tuple{\Rej,\TheNV}}$ implies
    $\DeffValueVar \in \BMatrixInterp{\DefSymbol}{\DefBMatrix}(\DefValueVar_1,\ldots,\DefValueVar_k)$;
    \item if $\nBAnaCalcCon{\IntermAccSet{\CutPropSet}}{\SetCompl{\IntermRejSet{\CutPropSet}}}{\IntermRejSet{\CutPropSet}}{\SetCompl{\IntermAccSet{\CutPropSet}}}{\BIndAxiomatRuleNonTotalName}{\GenSubfmlas{}{\Discriminator{\DefBMatrix}}}$,
    then
    $\Set{\DefValueVar \in \DefVSet \mid  \DiscriminatorVal{\DefCogVar}{\DefCogVar}{\DefValueVar}{\DefBMatrix}(\DefFm) \subseteq \CutPropSet_\DeffCogVar
    \text{ and } \DiscriminatorVal{\DefCogVar}{\InvCog{\DefCogVar}}{\DefValueVar}{\DefBMatrix}(\DefFm) \subseteq \SetCompl{\CutPropSet_\DeffCogVar},\text{for each }\\
    \text{\qquad\qquad\qquad\qquad\;} \Tuple{\DefCogVar,\DeffCogVar} \in \Set{\Tuple{\Acc,\TheV}, \Tuple{\Rej,\TheNV}}
    \text{ and } \DefFm \in \Subfmlas{\DefBseqVar}}
    \in
    \TotalSubMVSet{\DefBMatrix}$.
\end{enumerate}
\end{lemma}
\begin{proof}
We prove the contrapositive version of each
item below.
\begin{enumerate}
    \item Assume that for some $\DefFm \in \Subfmlas{\DefBseqVar}$
    there is no $\DefValueVar \in \DefVSet$ such that
    $\DiscriminatorVal{\DefCogVar}{\DefCogVar}{\DefValueVar}{\DefBMatrix}(\DefFm) \subseteq \CutPropSet_\DeffCogVar
    \text{ and } \DiscriminatorVal{\DefCogVar}{\InvCog{\DefCogVar}}{\DefValueVar}{\DefBMatrix}(\DefFm) \subseteq \SetCompl{\CutPropSet_\DeffCogVar}$,
    for each $\Tuple{\DefCogVar,\DeffCogVar} \in \Set{\Tuple{\Acc,\TheV}, \Tuple{\Rej,\TheNV}}$.
    Consider then the set
    \[
        \DefSetValueVarProp \SymbDef \Set{\DefValueVar \in \DefVSet \mid \DiscriminatorVal{\Acc}{\NAcc}{\DefValueVar}{\DefBMatrix}(\DefFm) \cap \CutPropSet_\TheV \neq \EmptySet
        \text{ or }
        \DiscriminatorVal{\Rej}{\NRej}{\DefValueVar}{\DefBMatrix}(\DefFm) \cap \CutPropSet_\TheNV \neq \EmptySet
        }
    \]
    and $\VSetComp{\DefSetValueVar} \SymbDef \SetDiff{\DefVSet}{\DefSetValueVarProp}$.
    Define, for each $\Tuple{\DefCogVar,\DeffCogVar} \in \Set{\Tuple{\Acc,\TheV}, \Tuple{\Rej,\TheNV}}$,
     the set $\DiscriminatorValCh{\DefCogVar}{\InvCog{\DefCogVar}}{\VSetComp{\DefSetValueVar}}{\DefBMatrix}$
    by choosing for each $\DefValueVar \in \DefSetValueVarProp$ a formula
    $\DefSeparator$ such that
    $\DefSeparator(\DefFm) \in \DiscriminatorVal{\DefCogVar}{\InvCog{\DefCogVar}}{\DefValueVar}{\DefBMatrix}(\DefFm) \cap \CutPropSet_\DeffCogVar$, when present.
    Similarly,
    define the set $\DiscriminatorValCh{\DefCogVar}{\DefCogVar}{\VSetComp{\DefSetValueVar}}{\DefBMatrix}$
    by choosing for each $\DefValueVar \in \VSetComp{\DefSetValueVar}$ a formula
    $\DefSeparator$ such that $\DefSeparator(\DefFm) \in \DiscriminatorVal{\DefCogVar}{\DefCogVar}{\DefValueVar}{\DefBMatrix}(\DefFm) \cap \SetCompl{\CutPropSet_\DeffCogVar}$, when present.
    Notice that the construction of $\DefSetValueVarProp$
    guarantees the existence of the pairs
    $\Tuple{\DiscriminatorValCh{\Acc}{\Acc}{\VSetComp{\DefSetValueVar}}{\DefBMatrix},\DiscriminatorValCh{\Rej}{\Rej}{\VSetComp{\DefSetValueVar}}{\DefBMatrix}}$ and $\Tuple{\DiscriminatorValCh{\Acc}{\NAcc}{\DefSetValueVarProp}{\DefBMatrix},\DiscriminatorValCh{\Rej}{\NRej}{\DefSetValueVarProp}{\DefBMatrix}}$.
    Since $\DiscriminatorValCh{\DefCogVar}{\DefCogVar}{\VSetComp{\DefSetValueVar}}{\DefBMatrix}(\DefFm) \subseteq \SetCompl{\CutPropSet_\DeffCogVar} \cap \GenSubfmlas{}{\Discriminator{\DefBMatrix}}(\DefBseqVar)$
    and
    $\DiscriminatorValCh{\DefCogVar}{\InvCog{\DefCogVar}}{\DefSetValueVarProp}{\DefBMatrix}(\DefFm) \subseteq \CutPropSet_\DeffCogVar \cap \GenSubfmlas{}{\Discriminator{\DefBMatrix}}(\DefBseqVar)$ for each $\Tuple{\DefCogVar,\DeffCogVar} \in \Set{\Tuple{\Acc,\TheV}, \Tuple{\Rej,\TheNV}}$, we have
    $\BAnaCalcCon{\IntermAccSet{\CutPropSet}}{\SetCompl{\IntermRejSet{\CutPropSet}}}{\IntermRejSet{\CutPropSet}}{\SetCompl{\IntermAccSet{\CutPropSet}}}{\BIndAxiomatRuleExistsName}{\GenSubfmlas{}{\Discriminator{\DefBMatrix}}}$.
    \item Suppose that for some
    $\DefFm \in \Subfmlas{\DefBseqVar}$
    and $\DefValueVar \in \DefVSet$ such that
    $\DiscriminatorVal{\DefCogVar}{\DefCogVar}{\DefValueVar}{\DefBMatrix}(\DefFm) \subseteq \CutPropSet_\DeffCogVar$
    and $\DiscriminatorVal{\DefCogVar}{\InvCog{\DefCogVar}}{\DefValueVar}{\DefBMatrix}(\DefFm) \subseteq \SetCompl{\CutPropSet_\DeffCogVar}$
    for each
    $\Tuple{\DefCogVar,\DeffCogVar} \in \Set{\Tuple{\Acc,\TheV}, \Tuple{\Rej,\TheNV}}$, we have either (a)
    $\DefValueVar \in \BMatrixDistSet{\DefffCogVar}{\DefBMatrix}$
    and $\DefFm \not\in \CutPropSet_\DeffffCogVar$ 
    (i.e. $\DefFm \in \SetCompl{\CutPropSet_\DeffffCogVar}$)
    or (b) $\DefValueVar \not\in \BMatrixDistSet{\DefffCogVar}{\DefBMatrix}$
    and $\DefFm \in \CutPropSet_\DeffffCogVar$
    , 
    for some
    $\Tuple{\DefffCogVar,\DeffffCogVar} \in \Set{\Tuple{\Acc,\TheV}, \Tuple{\Rej,\TheNV}}$.
    Notice that, for any
    $\Tuple{\DefCogVar,\DeffCogVar} \in \Set{\Tuple{\Acc,\TheV}, \Tuple{\Rej,\TheNV}}$,
    we have
        $\DiscriminatorVal{\DefCogVar}{\DefCogVar}{\DefValueVar}{\DefBMatrix}(\DefFm)  \subseteq \CutPropSet_\DeffCogVar \cap \GenSubfmlas{}{\Discriminator{\DefBMatrix}}(\DefBseqVar)$
    and
    $\DiscriminatorVal{\DefCogVar}{\InvCog{\DefCogVar}}{\DefValueVar}{\DefBMatrix}(\DefFm) \subseteq \SetCompl{\CutPropSet_\DeffCogVar} \cap \GenSubfmlas{}{\Discriminator{\DefBMatrix}}(\DefBseqVar)$,  
    implying that, in any of the
    cases (a) or (b),
    we have 
    $(\DiscriminatorVal{\DefffCogVar}{\DefffCogVar}{\DefValueVar}{\DefBMatrix}(\DefProp) \cup \FmlaPlacerByValue{\DefffCogVar}(\DefValueVar))(\DefFm) \subseteq \CutPropSet_\DeffffCogVar \cap \GenSubfmlas{}{\Discriminator{\DefBMatrix}}(\DefBseqVar)$
    and
    $(\DiscriminatorVal{\DefffCogVar}{\InvCog{\DefffCogVar}}{\DefValueVar}{\DefBMatrix}(\DefProp) \cup \FmlaPlacerByValue{\InvCog{\DefffCogVar}}(\DefValueVar))(\DefFm) \subseteq \SetCompl{\CutPropSet_\DeffffCogVar} \cap \GenSubfmlas{}{\Discriminator{\DefBMatrix}}(\DefBseqVar)$.
    Thus we have
    $\BAnaCalcCon{\IntermAccSet{\CutPropSet}}{\SetCompl{\IntermRejSet{\CutPropSet}}}{\IntermRejSet{\CutPropSet}}{\SetCompl{\IntermAccSet{\CutPropSet}}}{\BIndAxiomatRuleDName}{\GenSubfmlas{}{\Discriminator{\DefBMatrix}}}$ from an instance
    of one of the schemas
    of $\BIndAxiomatRuleDName$, 
    depending on the value of
    $\DefffCogVar$, obtained via
    a substitution mapping $\DefProp$ to $\DefFm$.
    \item \sloppy
    Suppose that there is a connective
$\DefSymbol \in \SigAritySet{\Sigma}{k}$,
a formula
    $\DefFm = \DefSymbol(\DefFm_1,\ldots,\DefFm_k) \in \Subfmlas{\DefBseqVar}$, a sequence $\Family{\DefValueVar_i}{i=1}^k$ of truth-values
    with
    $\DiscriminatorVal{\DefCogVar}{\DefCogVar}{\DefValueVar_i}{\DefBMatrix}(\DefFm_i) \subseteq \CutPropSet_\DeffCogVar$
    and $\DiscriminatorVal{\DefCogVar}{\InvCog{\DefCogVar}}{\DefValueVar_i}{\DefBMatrix}(\DefFm_i) \subseteq \SetCompl{\CutPropSet_\DeffCogVar}$ for
    each $1 \leq i \leq k$ and $\Tuple{\DefCogVar,\DeffCogVar} \in \Set{\Tuple{\Acc,\TheV}, \Tuple{\Rej,\TheNV}}$, and some $\DeffValueVar \not\in \BMatrixInterp{\DefSymbol}{\DefBMatrix}(\DefValueVar_1,\ldots,\DefValueVar_k)$
    such that
    $\DiscriminatorVal{\DefCogVar}{\DefCogVar}{\DeffValueVar}{\DefBMatrix}(\DefFm) \subseteq \CutPropSet_\DeffCogVar$
    and $\DiscriminatorVal{\DefCogVar}{\InvCog{\DefCogVar}}{\DeffValueVar}{\DefBMatrix}(\DefFm) \subseteq \SetCompl{\CutPropSet_\DeffCogVar}$ for each $\Tuple{\DefCogVar,\DeffCogVar} \in \Set{\Tuple{\Acc,\TheV}, \Tuple{\Rej,\TheNV}}$.
    Then
    $
        \bigcup_{1 \leq i \leq k}
        \DiscriminatorVal{\DefCogVar}{\DefCogVar}{\DefValueVar_i}{\DefBMatrix}(\DefFm_i) 
        \cup \DiscriminatorVal{\DefCogVar}{\DefCogVar}{\DeffValueVar}{\DefBMatrix}(\DefFm)
        \subseteq \CutPropSet_\DeffCogVar \cap 
        \GenSubfmlas{}{\Discriminator{\DefBMatrix}}(\DefBseqVar)
        \text{ and }
        \bigcup_{1 \leq i \leq k}
        \DiscriminatorVal{\DefCogVar}{\InvCog{\DefCogVar}}{\DefValueVar_i}{\DefBMatrix}(\DefFm_i) 
        \cup \DiscriminatorVal{\DefCogVar}{\InvCog{\DefCogVar}}{\DeffValueVar}{\DefBMatrix}(\DefFm)
        \subseteq \SetCompl{\CutPropSet_\DeffCogVar} \cap 
        \GenSubfmlas{}{\Discriminator{\DefBMatrix}}(\DefBseqVar)
    $
    for each $\Tuple{\DefCogVar,\DeffCogVar} \in \Set{\Tuple{\Acc,\TheV}, \Tuple{\Rej,\TheNV}}$, and thus
    we have
    $\BAnaCalcCon{\IntermAccSet{\CutPropSet}}{\SetCompl{\IntermRejSet{\CutPropSet}}}{\IntermRejSet{\CutPropSet}}{\SetCompl{\IntermAccSet{\CutPropSet}}}{\BIndAxiomatRuleSigmaName}{\GenSubfmlas{}{\Discriminator{\DefBMatrix}}}$.
    \item Let
    $\DefSetValueVar = \Set{\DefValueVar \in \DefVSet \mid \DiscriminatorVal{\DefCogVar}{\DefCogVar}{\DefValueVar}{\DefBMatrix}(\DefFm) \subseteq \CutPropSet_\DeffCogVar
    \text{ and } \DiscriminatorVal{\DefCogVar}{\InvCog{\DefCogVar}}{\DefValueVar}{\DefBMatrix}(\DefFm) \subseteq \SetCompl{\CutPropSet_\DeffCogVar},
    \text{ for each } \Tuple{\DefCogVar,\DeffCogVar} \in \Set{\Tuple{\Acc,\TheV}, \Tuple{\Rej,\TheNV}}
    \text{ and } \DefFm \in \Subfmlas{\DefBseqVar}}$.
    For each 
    $\DefValueVar \in \DefSetValueVar$,
    pick a formula $\DefFm_\DefValueVar \in \GenSubfmlas{}{\Discriminator{\DefBMatrix}}(\DefBseqVar)$
    such that 
    $\DiscriminatorVal{\DefCogVar}{\DefCogVar}{\DefValueVar}{\DefBMatrix}(\DefFm_\DefValueVar) \subseteq \CutPropSet_\DeffCogVar
    \text{ and } \DiscriminatorVal{\DefCogVar}{\InvCog{\DefCogVar}}{\DefValueVar}{\DefBMatrix}(\DefFm_\DefValueVar) \subseteq \SetCompl{\CutPropSet_\DeffCogVar},
    \text{for each } \Tuple{\DefCogVar,\DeffCogVar} \in \Set{\Tuple{\Acc,\TheV}, \Tuple{\Rej,\TheNV}}$.
    Easily, then, for each $ \Tuple{\DefCogVar,\DeffCogVar} \in \Set{\Tuple{\Acc,\TheV}, \Tuple{\Rej,\TheNV}}$, we have
    $
    \bigcup_{\DefValueVar \in \DefSetValueVar}
    \DiscriminatorVal{\DefCogVar}{\DefCogVar}{\DefValueVar}{\DefBMatrix}(\DefFm_\DefValueVar) \subseteq \CutPropSet_\DeffCogVar
    \cap 
    \GenSubfmlas{}{\Discriminator{\DefBMatrix}}(\DefBseqVar)
    \text{ and }
     \bigcup_{\DefValueVar \in \DefSetValueVar}
    \DiscriminatorVal{\DefCogVar}{\InvCog{\DefCogVar}}{\DefValueVar}{\DefBMatrix}(\DefFm_\DefValueVar) \subseteq \SetCompl{\CutPropSet_\DeffCogVar}
    \cap 
    \GenSubfmlas{}{\Discriminator{\DefBMatrix}}(\DefBseqVar)   
    $,
    and so 
    $\BAnaCalcCon{\IntermAccSet{\CutPropSet}}{\SetCompl{\IntermRejSet{\CutPropSet}}}{\IntermRejSet{\CutPropSet}}{\SetCompl{\IntermAccSet{\CutPropSet}}}{\BIndAxiomatRuleNonTotalName}{\GenSubfmlas{}{\Discriminator{\DefBMatrix}}}$
    if
    $\DefSetValueVar \not\in \TotalSubMVSet{\DefBMatrix}$.
\end{enumerate}
\end{proof}

\begin{theorem}
\label{the:completeness}
If $\Discriminator{\DefBMatrix}$ is a discriminator for
a $\theBPN[\Sigma]$--matrix
$\DefBMatrix$, then
the calculus
$\BIndAxiomatName{\DefBMatrix}{\Discriminator{\DefBMatrix}}{}$
is complete with respect to
$\DefBMatrix$ and $\SeparatorsSet{\DefBMatrix}$--analytic.
\end{theorem}
\begin{proof}
\vspace{-2mm}
Let $\DefBseqVar \SymbDef \BStat{\DefAccSet}{\DefNRejSet}{\DefRejSet}{\DefNAccSet}$ 
be a \theB-statement
and
suppose that
(a) $\nBAnaCalcCon{\DefAccSet}{\DefNRejSet}{\DefRejSet}{\DefNAccSet}{\DefCalculusName}{\GenSubfmlas{}{\Discriminator{\DefBMatrix}}}$. Our goal is to build an
$\DefBMatrix$-valuation witnessing
$\nBEnt{\DefAccSet}{\DefNRejSet}{\DefRejSet}{\DefNAccSet}{\DefBMatrix}$.
From (a), by $\ref{prop:BConC}$,
we have that (b) there are
$\DefAccSet \subseteq \IntermAccSet{\CutPropSet} \subseteq \SetCompl{\DefNAccSet}$
and
$\DefRejSet \subseteq \IntermRejSet{\CutPropSet} \subseteq \SetCompl{\DefNRejSet}$ such that
$\nBAnaCalcCon{\IntermAccSet{\CutPropSet}}{\SetCompl{\IntermRejSet{\CutPropSet}}}{\IntermRejSet{\CutPropSet}}{\SetCompl{\IntermAccSet{\CutPropSet}}}{\DefCalculusName}{\GenSubfmlas{}{\Discriminator{\DefBMatrix}}}$. Consider then a mapping
$f : \Subfmlas{\DefBseqVar} \to \DefVSet$
with (c)
$f(\DefFm) \in \BMatrixDistSet{\DefCogVar}{\DefBMatrix}$
    if{f} $\DefFm \in \CutPropSet_\DeffCogVar$, for
    $\Tuple{\DefCogVar,\DeffCogVar} \in \Set{\Tuple{\Acc,\TheV}, \Tuple{\Rej,\TheNV}}$,
    whose existence is guaranteed by
    items (1) and (2) of
    Lemma~\ref{lem:auxcompleteness}.
Notice that items (3) and (4)
of this same proposition imply, respectively, that $f(\DefSymbol(\DefFm_1, \ldots, \DefFm_k)) \in \BMatrixInterp{\DefSymbol}{\DefBMatrix}(
f(\DefFm_1), \ldots, f(\DefFm_k))$
for every $\DefSymbol(\DefFm_1, \ldots, \DefFm_k) \in \GenSubfmlas{}{\Discriminator{\DefBMatrix}}(\DefBseqVar)$,
and $f(\Subfmlas{\DefBseqVar}) \in \TotalSubMVSet{\DefBMatrix}$.
Hence, $f$ may be extended to 
an $\DefBMatrix$-valuation $\DefVal$
and, from (b) and (c), we have
$\DefVal(\CogSet{\DefFmlaSet}{\DefCogVar})
        \subseteq \BMatrixDistSet{\DefCogVar}{\DefBMatrix}$
        for each $\DefCogVar \in \CogsSet$,
so $\nBEnt{\DefAccSet}{\DefNRejSet}{\DefRejSet}{\DefNAccSet}{\DefBMatrix}$.
\end{proof}


\begin{lemma}
\label{lem:correctness}
Let
$R$ be a finite set of finitary rule instances. 
Then
the procedure 
$\ProofSAlgoName(
\Tuple{\DefAccSet, \DefRejSet},
\Tuple{\DefNAccSet, \DefNRejSet}, 
R)$
always terminates, returning~a tree 
that is $\Tuple{\DefNAccSet, \DefNRejSet}$-closed
iff
$\BRuleInstCon{\DefAccSet}{\DefNRejSet}{\DefRejSet}{\DefNAccSet}{R}$.
\end{lemma}
\begin{proof}
\sloppy
Let $R$ be a finite set of finitary rule instances, 
set $F\SymbDef\Tuple{\DefAccSet, \DefRejSet}$ and set $C\SymbDef\Tuple{\DefNRejSet, \DefNAccSet}$. 
We proceed by induction on $ \SetSize{R}$. In the base case, $R=\EmptySet$, the algorithm obviously terminates
and returns a proof of $\DefBseqVar$ iff $\BRuleInstCon{\DefAccSet}{\DefNRejSet}{\DefRejSet}{\DefNAccSet}{\EmptySet}$.
In the inductive step, assume that $\SetSize{R} \geq 1$ and that (IH): the present
lemma holds for all sets of rule instances $R'$
with $\SetSize{R'} = \SetSize{R} - 1$. Since $R$
is finite and contains only finitary rule instances,
and each recursive call (line \ref{algoline:reccall}) terminates by (IH), the whole algorithm
terminates. Also, if a $C$-closed tree is produced,
it means that one of the conditions
in lines \ref{algoline:base}, \ref{algoline:starcond} or \ref{algoline:closedcond}
was satisfied. The first possibility (line \ref{algoline:base}) was treated in
the base case. The second one (line \ref{algoline:starcond}) means that
there is a rule instance in $R$ with
an empty succedent satisfying the antecedents,
in which case a tree with its root labelled
with $F$ having a single child labelled with
$\StarLabel$ is returned, clearly 
bearing witness to
$\BRuleInstCon{\DefAccSet}{\DefNRejSet}{\DefRejSet}{\DefNAccSet}{R}$.
The third possibility (line \ref{algoline:closedcond})
means that there is a rule instance
$\DefRuleInstVar \SymbDef \TwoDRule{\DeffffAccSet}{\DeffffNRejSet}{\DeffffRejSet}{\DeffffNAccSet} \in R$ applicable to
the antecedents in $F$ (line \ref{algoline:condpremiss}), and, by (IH), the recursive calls (line \ref{algoline:reccall}) produce
trees that bear witness to
$\BRuleInstCon{\DefAccSet,\DefFm}{\DefNRejSet}{\DefRejSet}{\DefNAccSet}{\SetDiff{R}{\Set{\DefRuleInstVar}}}$ and
$\BRuleInstCon{\DefAccSet}{\DefNRejSet}{\DefRejSet, \DeffFm}{\DefNAccSet}{\SetDiff{R}{\Set{\DefRuleInstVar}}}$ for each $\DefFm \in \DeffffFmlaSet_\NAcc$ and
$\DeffFm \in \DeffffFmlaSet_\NRej$.
The resulting tree, then, bears witness to $\BRuleInstCon{\DefAccSet}{\DefNRejSet}{\DefRejSet}{\DefNAccSet}{R}$.
On the other hand, if an open
tree is produced, then for a rule instance
$\DefRuleInstVar \SymbDef \TwoDRule{\DeffffAccSet}{\DeffffNRejSet}{\DeffffRejSet}{\DeffffNAccSet} \in R$ applicable to $F$, some
recursive call resulted in an open tree. 
Assume, without loss of generality, that such call referred to an expansion by $\DefFm \in \DeffffNAccSet$. Then, by (IH),
$\nBRuleInstCon{\DefAccSet,\DefFm}{\DefNRejSet}{\DefRejSet}{\DefNAccSet}{\SetDiff{R}{\Set{\DefRuleInstVar}}}$. Because $\DefFm \in \DeffffNAccSet$,
the instance
$\DefRuleInstVar$ does not play any role in
deriving $\BStat{\DefAccSet,\DefFm}{\DefNRejSet}{\DefRejSet}{\DefNAccSet}$, so
we have $\nBRuleInstCon{\DefAccSet,\DefFm}{\DefNRejSet}{\DefRejSet}{\DefNAccSet}{R}$ and, by \ref{prop:BConD}, it follows that
$\nBRuleInstCon{\DefAccSet}{\DefNRejSet}{\DefRejSet}{\DefNAccSet}{R}$.
\end{proof}

\begin{lemma}
\label{lem:anacalcproofs}
If $\DefCalculusName$ is $\DeffFmlaSet$-analytic, then
$\ProofSAlgoName$ is a proof-search algorithm for $\DefCalculusName$. 
\end{lemma}
\begin{proof}
We know that $\RuleInstSetCalcSeq{\DefCalculusName}{\DefBseqVar}$
must be enough to provide a derivation
of~$\DefBseqVar$, since $\DefCalculusName$ is $\DeffFmlaSet$-analytic.
Clearly, such set is finite and contains only
finitary rule instances, hence
the present result is a direct consequence of
Lemma \ref{lem:correctness}. 
\end{proof}

\begin{lemma}
\label{lem:complexity}
The worst-case running time of $\ProofSAlgoName(
\Tuple{\DefAccSet, \DefRejSet},
\Tuple{\DefNAccSet, \DefNRejSet}, 
R)$ is $O(b^{n} + n \cdot \Polynomial{}(s))$.
\end{lemma}
\begin{proof}
The worst-case running-time $T(n, s)$ of $\ProofSAlgoName$ occurs
when $\BRuleInstCon{\DefAccSet}{\DefNRejSet}{\DefRejSet}{\DefNAccSet}{R}$,
 the set $R$ needs to be entirely inspected until an applicable rule instance is found, and such an instance does not have an empty set of succedents. The following table
assigns a cost and simplified upper bounds to execution times 
of the relevant instructions of
Algorithm~\ref{algo:proofsearch} in the described
scenario.
\begin{table}
    \centering
    \begin{tabular}{c|c|c}
        \toprule
         Instruction line & Cost & Times \\
         \midrule
         \ref{algoline:maket} & $c_1$ & 1\\
         \ref{algoline:base} & $\Polynomial{}(s)$ & 1\\ 
         \ref{algoline:looprules} & $c_2$ & $n$\\ 
         \ref{algoline:condpremiss} & $\Polynomial{}(s)$ & $n$\\ 
         \ref{algoline:starcond} & $c_3$ & $1$\\ 
         \ref{algoline:loopfmlas} & $c_4$ & $b$\\ 
         \ref{algoline:reccall} & $T(n-1, s + \Polynomial{}(s) )$ & $b$\\
         \ref{algoline:addchild} & $c_5$ & $b$\\ 
         \ref{algoline:notcclosed} & $\Polynomial{}(s)$ & $b$\\ 
         \ref{algoline:closedcond} & $\Polynomial{}(s)$ & 1
    \end{tabular}
\end{table}

\noindent 
Notice that $T(0,s) = c_1 + \Polynomial{}(s)$ and, based on the
assignments above and after some algebraic manipulations, we have, for $n \geq 1$,
\begin{equation}
    T(n, s) \leq b \cdot T(n-1, s + \Polynomial{}(s)) + 2n \cdot \Polynomial{}(s).
\end{equation}

\noindent We prove by induction on $n$
that $T(n,s) \in O(b^n + n \cdot \Polynomial{}(s))$.
We will take advantage of the asymptotic notation
and choose the base case as $n=1$. In such case, we have $T(1, s) \leq  bT(0, s+\Polynomial{}(s)) + 2\Polynomial{}(s) = 2\Polynomial{}(s) + bc_1 + b\Polynomial{}(s) = c_1b + (2+b)\Polynomial{}(s)$, and the upper bound suffices.
In the inductive step, let $n > 1$ and
assume, for all $s \geq 0$, that
$T(n-1, s) \leq k_1 \cdot b^{n-1} + k_1 \cdot (n-1) \cdot \Polynomial{}(s)$, for some $k_1 > 0$.
Then
\begin{align*}
    T(n,s) & \leq b \cdot T(n-1, s + \Polynomial{}(s)) + 2n \cdot \Polynomial{}(s)\\
    & \leq b \cdot (k_1 \cdot b^{n-1} + k_1 \cdot (n-1) \cdot \Polynomial{}(s)) + 2n \cdot \Polynomial{}(s)\\
    & = k_1 \cdot b^n + b\cdot k_1 \cdot (n-1) \cdot \Polynomial{}(s) + 2n \cdot \Polynomial{}(s)\\
    &= k_1 \cdot b^n + (2n + b \cdot k_1 \cdot n - b \cdot k_1) \cdot \Polynomial{}(s)\\
    &\leq k_1 \cdot b^n + (2n + b \cdot k_1 \cdot n) \cdot \Polynomial{}(s)\\
    &= k_1 \cdot b^n + k_2 \cdot n \cdot \Polynomial{}(s), \text{ with $k_2 = 2 + b \cdot k_1$}\\
    &\leq k_3 \cdot (b^n + n \cdot \Polynomial{}(s)), \text{ with } k_3 = \max\{k_1,k_2\}\\
    &\in O(b^n + n \cdot \Polynomial{}(s))
\end{align*}
\end{proof}

\begin{theorem}
\label{the:analyticexp}
If $\DefCalculusName$ is $\DeffFmlaSet$-analytic, $\ProofSAlgoName$
is a proof-search 
for $\DefCalculusName$
that is exponential time in general, and is 
polynomial time
if $\DefCalculusName$ contains only rules with at most one succedent.
\end{theorem}
\begin{proof}
\sloppy 
Clearly, the set of all instances
of rules of $\DefCalculusName$
using only 
formulas in 
$\GenSubfmlas{}{\DeffFmlaSet}(\DefBseqVar)$ is finite and contains only
finitary rule instances, and its size is polynomial in $\SeqSizeSymb(\DefBseqVar)$. 
The announced result then follows
directly from Lemma \ref{lem:complexity}.
\end{proof}

\begin{theorem}
\label{the:coNPdec}
If $\DefCalculusName$ is $\DeffFmlaSet$-analytic, then the problem of deciding $\BCalcConName{\DefCalculusName}$ is in $\coNP$.
\end{theorem}
\begin{proof}
Let $\DefBseqVar \SymbDef \BStat{\DefAccSet}{\DefNRejSet}{\DefRejSet}{\DefNAccSet}$.
Given a pair $\Tuple{\DeffAccSet, \DeffRejSet}$ with $\DeffAccSet \cup \DeffRejSet \subseteq \GenSubfmlas{}{\DeffFmlaSet}(\DefBseqVar)$, $\DefFmlaSet_\DefCogVar \subseteq \DeffFmlaSet_\DefCogVar$
and $\DefFmlaSet_{\InvCog{\DefCogVar}} \cap \DeffFmlaSet_{\DefCogVar} = \EmptySet$ 
for each $\DefCogVar \in \Set{\Acc,\Rej}$, 
if we check that for every applicable rule instance
$\TwoDRule{\DeffffAccSet}{\DeffffNRejSet}{\DeffffRejSet}{\DeffffNAccSet} \in \RuleInstSetCalcSeq{\DefCalculusName}{\DefBseqVar}$
we have
$\DeffFmlaSet_{\DefCogVar} \cap \DeffffFmlaSet_{\InvCog{\DefCogVar}} \neq \EmptySet$
for each $\DefCogVar \in \Set{\Acc,\Rej}$,
then $\nBAnaCalcCon{\DeffAccSet}{\SetDiff{\GenSubfmlas{}{\DeffFmlaSet}(\DefBseqVar)}{\DeffRejSet}}{\DeffRejSet}{\SetDiff{\GenSubfmlas{}{\DeffFmlaSet}(\DefBseqVar)}{\DeffAccSet}}{\DefCalculusName}{\GenSubfmlas{}{\DeffFmlaSet}}$,
and thus, by \ref{prop:BConD},
$\nBAnaCalcCon{\DefAccSet}{\DefNRejSet}{\DefRejSet}{\DefNAccSet}{\DefCalculusName}{\GenSubfmlas{}{\DeffFmlaSet}}$.
Since the amount of rule instances is
polynomial in the size of $\DefBseqVar$,
by guessing in polynomial time the latter
pair and performing the described test we obtain a polynomial-time non-deterministic algorithm
to verify if $\nBAnaCalcCon{\DefAccSet}{\DefNRejSet}{\DefRejSet}{\DefNAccSet}{\DefCalculusName}{\GenSubfmlas{}{\DeffFmlaSet}}$, and so the problem of deciding such calculus is in $\coNP$.
\end{proof}
